\documentclass[journal,12pt,onecolumn,draftclsnofoot,]{IEEEtran}

\usepackage[T1]{fontenc}
\usepackage{subcaption}
\usepackage{dsfont}
\usepackage{color}
\usepackage{varioref}
\usepackage{textcomp}
\usepackage{amsthm}
\usepackage{amsmath}
\usepackage{amssymb}
\usepackage{graphicx}
\usepackage{setspace}
\usepackage{esint}
\usepackage{algorithm}
\usepackage{algpseudocode}
\usepackage{pifont}
\usepackage{wrapfig}
\usepackage{multirow}
\usepackage{tabu}
\usepackage{amsmath}

\usepackage{enumitem}
\usepackage{cite}
\usepackage{cleveref}

\makeatletter

\newcommand{\Rmnum}[1]{\expandafter\@slowromancap\romannumeral #1@}
\makeatother

\newtheorem{theorem}{Theorem}
\newtheorem{lemma}{Lemma}
\newtheorem{remark}{Remark}

\newtheorem{assumption}{Assumption}
\newtheorem{corollary}{Corollary}
\theoremstyle{definition}

\providecommand{\propositionname}{Proposition}

\ifCLASSINFOpdf

\else

\fi

\usepackage[T1]{fontenc}
\usepackage{etoolbox}

\makeatletter
\patchcmd{\maketitle}{\@fnsymbol}{\@alph}{}{}  
\makeatother

\title{Convergence of Update Aware Device Scheduling for Federated Learning at the Wireless Edge}
\author{\IEEEauthorblockN{Mohammad Mohammadi Amiri, Deniz G\"und\"uz, Sanjeev R. Kulkarni,\\ H. Vincent Poor\thanks{M. Mohammadi Amiri, S. R. Kulkarni, and H. V. Poor are with the Department of Electrical Engineering, Princeton University, Princeton, NJ 08544, USA (e-mail: \{mamiri, kulkarni, poor\}@princeton.edu).} \thanks{D. G\"und\"uz is with the Department of Electrical and Electronic Engineering, Imperial College London, London SW7 2AZ, U.K. (e-mail: d.gunduz@imperial.ac.uk).}
}
}

\begin{document}

\maketitle

\begin{abstract}
We study federated learning (FL) at the wireless edge, where power-limited devices with local datasets collaboratively train a joint model with the help of a remote parameter server (PS). We assume that the devices are connected to the PS through a bandwidth-limited shared wireless channel. At each iteration of FL, a subset of the devices are scheduled to transmit their local model updates to the PS over orthogonal channel resources, while each participating device must compress its model update to accommodate to its link capacity. 
We design novel scheduling and resource allocation policies that decide on the subset of the devices to transmit at each round, and how the resources should be allocated among the participating devices, not only based on their channel conditions, but also on the significance of their local model updates. 
We then establish convergence of a wireless FL algorithm with device scheduling, where devices have limited capacity to convey their messages. 
The results of numerical experiments show that the proposed scheduling policy, based on both the channel conditions and the significance of the local model updates, provides a better long-term performance than scheduling policies based only on either of the two metrics individually. Furthermore, we observe that when the data is independent and identically distributed (i.i.d.) across devices, selecting a single device at each round provides the best performance, while when the data distribution is non-i.i.d., scheduling multiple devices at each round improves the performance. This observation is verified by the convergence result, which shows that the number of scheduled devices should increase for a less diverse and more biased data distribution.\makeatletter{\renewcommand*{\@makefnmark}{}\footnotetext{This work was supported in part by the U.S. National Science Foundation under Grant CCF-0939370, and by the European Research Council (ERC) Starting Grant BEACON (grant agreement no. 677854).}\makeatother}    
\end{abstract}


\section{Introduction}\label{SecIntro}

Devices at the wireless network edge generate a huge amount of data that can be exploited to make sense of the state of a system or to predict its future states. Internet-of-things (IoT) devices, unmanned aerial vehicles (UAVs), or extended reality (XR) technologies are prime examples, where data from multiple sensors must be continuously collected and processed. Many machine learning (ML) algorithms have been developed to exploit these massive datasets, but most current ML solutions focus on centralized algorithms, where a cloud server collects all the data to train a powerful model. However, offloading such massive amounts of data from the edge devices to the could server is often not feasible due to latency, bandwidth, or power constraints, or not allowed due to privacy concerns. A more desirable and practically viable approach is \textit{federated learning} (FL), which enables ML at the wireless edge while the data never leaves the edge devices.


FL utilizes the computational capabilities of edge devices to process their local datasets and collaboratively train a learning model with the help of a parameter server (PS) collecting the local model updates from the devices \cite{DCKonecnyFederated}. Due to unreliable links from the edge devices to the PS with limited energy and bandwidth, it is essential to develop distributed learning techniques with limited communication requirements \cite{DCKonecnyFederated,McMahan2017CommunicationEfficientLO,GoogleMcMahanFed,KonecnyRandDistMean,SmithFedMultiTask,KonecnyFLBeyondData,NishioFLClientSelecHetRes}. However, most of the existing works on communication-efficient FL ignore the physical layer characteristics, and consider perfect rate-limited links between the devices and the PS.


There have been few recent studies on FL taking into account the physical layer aspects of the wireless network from the devices to the PS. 
In \cite{AccelDNNFLEdge}, optimization over batch size and communication resources is proposed to speed up FL over wireless networks.
FL over a Gaussian multiple access channel (MAC) with limited bandwidth is studied in \cite{MohammadDenizDSGDCS}, and novel digital and analog communication approaches have been proposed for the transmissions from the devices. While gradient compression followed by digital channel coding is employed for the former, the latter analog approach benefits from the superposition property of the wireless channel by transmitting the gradients in an analog form. Bandwidth efficiency is achieved in \cite{MohammadDenizDSGDCS} by sparsification of the gradients and random linear projection. In \cite{KaibinParallelWork}, FL over broadband wireless fading MAC is studied, where analog transmission is combined with power control, and the devices perform channel inversion with the full knowledge of channel state information (CSI) to align their transmissions at the PS. In \cite{MohamamdDenizFLOverAirSPAWC19,FLTWCMohammadDenizFading}, the bandwidth-efficient analog communication technique in \cite{KaibinParallelWork} is combined with power control over a bandwidth-limited fading MAC, significantly reducing the communication load. Beamforming techniques at a multi-antenna PS for increasing the number of participating devices and overcoming the lack of CSI at the devices are introduced in \cite{YangFedLearOverAirComp} and \cite{MohammadTolgaDenizFLGlobalSIP}, respectively. In \cite{ChenVinceJointLearCommFL}, resource allocation across devices for FL over wireless channels is formulated as an optimization problem aiming to minimize the learning empirical loss function.
Frequency of participation of the devices is introduced as a device scheduling metric in \cite{YangArafaVinceAgeBasedFL}.
Also, a device scheduling policy for FL over wireless channels is studied in \cite{ShiZhouFastConverFL} to minimize the training delay. 
In \cite{DinhFLWireNetConver} the authors study resource allocation for FL over wireless networks aiming to minimize the total energy consumption, and in \cite{HowardVinceSchedulingsFL} performance analysis of FL under different scheduling policies, including random scheduling, round-robin scheduling and scheduling based on the link qualities of devices, is provided.   
A channel aware quantization technique is proposed in \cite{RaviCommEffFLGaussian} for digital transmission of the gradients over a Gaussian MAC with limited link capacities. 
Also, in \cite{DenizOneBitAgg} the authors develop a digital transmission utilizing the over-the-air aggregation feature of the wireless MAC from the devices to the PS.

In this paper, we build upon our previous work \cite{MohammadDenizSanjVinceISIT20}, and consider FL with digital transmission from the edge devices to the PS over a block fading wireless network with limited power and bandwidth resources. We design novel device scheduling and resource allocation techniques in order to select the devices that send their model updates at each iteration, and to allocate the limited wireless resources among the participating devices to perform orthogonal (interference-free) transmission. Specifically, we take into account the channel conditions and the significance of the local model updates at the devices for device scheduling to make sure that the channel resources are allocated to the devices with both important messages and sufficient-quality link capacity to convey their messages. Next we establish a convergence result for FL with device scheduling, where each participating device can only transmit a limited number of bits imposed by its channel capacity. To the best of our knowledge, this is the first convergence result evaluating the performance of FL as a function of the number of scheduled devices at each round, as well as the number of bits each participating device can transmit, which reduces with the number of participating devices.
The authors in \cite{XLiFedAveFLnonIID} develop a convergence result for FL with device scheduling, however the assumption is that devices can communicate with the PS over an interference-and error-free shared link with unlimited capacity.
Numerical results illustrate the advantages of considering both the channel conditions and the local model updates at the devices for device scheduling over scheduling based on either of the two metrics individually. 
Experiments on MNIST dataset illustrate that allocating all the resources to a single device offers the best performance when the data is independent and identically distributed (i.i.d.) across the devices, while scheduling more devices provides the best performance when data distribution is non-i.i.d.. The convergence result corroborates this observation. Modelling the asymmetry in the data distribution across devices by an increase in the variance of the local gradients, we observe that the convergence speed increases withe the number of scheduled devices when the variance of the gradients is large, while scheduling a single device provides the best performance when the variance is limited.


\textit{Notation}: We denote the set of real, natural and complex numbers by $\mathbb{R}$, $\mathbb{N}$ and $\mathbb{C}$, respectively. We let $[i] \triangleq \{ 1, \dots, i \}$. An $i$-length vector of all-zero entries is denoted by $\boldsymbol{0}_i$. We denote a circularly symmetric complex Gaussian distribution with real and imaginary components with variance $\sigma^2/2$ by $\mathcal{C N} \left( 0,\sigma^2 \right)$. Notation $\left| \cdot \right|$ represents the cardinality of a set or the magnitude of a complex value, the $l_2$-norm of a vector $\boldsymbol{x}$ is denoted by $\left\| \boldsymbol{x} \right\|_2$, and $\langle \boldsymbol{x}, \boldsymbol{y} \rangle$ denotes the inner product of vectors $\boldsymbol{x}$ and $\boldsymbol{y}$. For a set $\cal S$ of real numbers, ${\max}_{[K]} \mathcal S$ returns the $K$-element subset of $\mathcal S$ with the highest values.



\section{System Model}\label{SecProbFormul}

We consider FL across $M$ wireless devices, collaboratively training a model parameter vector $\boldsymbol{\theta} \in \mathbb{R}^d$ with the help of a remote parameter server (PS), to which they are connected through a shared wireless medium, to minimize an empirical loss function 
\begin{align}\label{GenEmpLossFunc}
F \left( \boldsymbol{\theta} \right) = \frac{1}{M} \sum\nolimits_{m=1}^{M} F_m \left(\boldsymbol{\theta} \right),   
\end{align}
where $F_m \left(\boldsymbol{\theta} \right)$ denotes the loss function at device $m$, $m \in [M]$.   

\subsection{FL System}\label{SubSecFL}
In FL, each device performs multiple \textit{stochastic gradient descent} (SGD) updates to minimize an empirical loss function with respect to its local dataset based on a globally consistent model parameter vector received from the PS. 
Let $\mathcal{B}_{m}$ denote the local dataset at device $m$, $m \in [M]$, with $\left| \mathcal{B}_{m} \right| = B$. The loss function at device $m$ is given by 
\begin{align}\label{LossFuncDevicem}
F_m \left( \boldsymbol{\theta} \right) = \frac{1}{B} \sum\nolimits_{\boldsymbol{u} \in \mathcal{B}_m} f \left(\boldsymbol{\theta}, \boldsymbol{u} \right), \quad m \in [M],   
\end{align}
where $f \left(\boldsymbol{\theta}, \boldsymbol{u} \right)$ is an empirical loss function defined by the learning task and quantifies the loss of model $\boldsymbol{\theta}$ at sample $\boldsymbol{u}$. During the $t$-th global iteration, $t=0, 1, \dots$, having received global model parameter vector $\boldsymbol{\theta} (t)$ from the PS, device $m$ performs a $\tau$-step local SGD, for some $\tau \in \mathbb{N}$. The $i$-th step of local SGD for the $t$-th global iteration at device $m$, $m \in [M]$, corresponds to the following update:
\begin{align}\label{ithStepSGDDevicem}
\boldsymbol{\theta}_m^{i+1} (t) = \boldsymbol{\theta}_m^i (t) - \eta^i_m (t) \nabla F_m \left( \boldsymbol{\theta}_m^i (t), \xi_m^i (t) \right),  \quad \mbox{$i \in [\tau]$},   
\end{align}
where $\eta^i_m (t)$ denotes the learning rate, and $\nabla F_m \left( \boldsymbol{\theta}_m^i (t), \xi_m^i (t)  \right)$ denotes the gradient computed at the local mini-batch sample $\xi_m^i (t)$, which is chosen uniformly at random from the local dataset, with $\mathbb{E}_{\xi} \left[ \nabla F_m \left( \boldsymbol{\theta}_m^i (t), \xi_m^i (t) \right) \right] = \nabla F_m \left( \boldsymbol{\theta}_m^i (t)  \right)$, where $\mathbb{E}_{\xi}$ denotes expectation with respect to the randomness of the stochastic gradient function. We set $\boldsymbol{\theta}_m^1 (t) = \boldsymbol{\theta} (t)$, $\forall m \in [M]$, i.e., the first local update is carried out on the global model parameter vector $\boldsymbol{\theta} (t)$ received from the PS. We further denote the model parameter vector after the $\tau$-th local update at device $m$ by $\boldsymbol{\theta}_m (t+1)$, i.e., $\boldsymbol{\theta}_m (t+1) = \boldsymbol{\theta}_m^{\tau+1} (t)$, $m \in [M]$.

Each device conveys its local model estimate after $\tau$ local steps to the PS, which updates the global model parameter vector $\boldsymbol{\theta} (t+1)$ by averaging these results
\begin{align}\label{GlobModelParUpdatePS}
\boldsymbol{\theta} (t+1) = \frac{1}{M} \sum\nolimits_{m=1}^{M} \boldsymbol{\theta}_m (t+1).  
\end{align}
This updated vector is then shared with the devices for further computations until convergence. Having defined the local model update at device $m$ as   
\begin{align}\label{DevicemModelDif}
\Delta \boldsymbol{\theta}_m (t) \triangleq \boldsymbol{\theta}_m (t+1) - \boldsymbol{\theta} (t), \quad m \in [M],
\end{align}
the update in \eqref{GlobModelParUpdatePS} corresponds to 
\begin{align}\label{PSGlobalUpdateModelDif}
\boldsymbol{\theta} (t+1) = \boldsymbol{\theta} (t) + \frac{1}{M} \sum\nolimits_{m=1}^{M} \Delta \boldsymbol{\theta}_m (t).
\end{align}
Thus, the PS requires model updates estimated at the devices to update the global model parameter vector.


When FL is implemented over a shared wireless medium, it is not reasonable to expect all the devices to be able to convey their model updates to the PS reliably, due to power and bandwidth constraints. When the available channel resources are shared between the devices, each device would be allocated only a limited bandwidth, and the amount of information that can be conveyed to the PS can be further limited due to fading. Hence, it is reasonable to schedule only a subset of the devices at each iteration so that the scheduled devices can convey their model updates with sufficient quality and reliability \cite{FLTWCMohammadDenizFading,YangArafaVinceAgeBasedFL,ShiZhouFastConverFL,HowardVinceSchedulingsFL}.

In this paper, we consider scheduling a $K$-element subset of the devices, denoted by $\mathcal{M} (t) \subset [M]$, $\left| \mathcal{M} (t) \right| = K$, at each global iteration step $t$, for the most efficient utilization of the power and bandwidth resources.
The PS determines the set of scheduled devices, and informs them for transmission at each round. 
Accordingly, the PS updates the global model parameter vector based on the received local model updates from only the scheduled devices as follows:
\begin{align}\label{PSGlobalUpdateModelDifScheduled}
\boldsymbol{\theta} (t+1) = \boldsymbol{\theta} (t) + \frac{1}{K} \sum\nolimits_{m \in \mathcal{M} (t)} \Delta \boldsymbol{\theta}_m (t).
\end{align}

\subsection{Wireless Medium}\label{SubSecWireless}

We consider a wireless medium with limited bandwidth from the devices to the PS to transmit the local model updates $\Delta \boldsymbol{\theta}_m (t) \in \mathbb{R}^d$, $\forall m \in [M]$. We assume a single-carrier block fading wireless channel with $n$ symbols (time slots) using TDMA for transmission from the devices to the PS\footnote{The single-carrier assumption is for ease of presentation, and the results in this paper can be extended to multi-carrier systems.}. We denote the length-$n$ input to the channel at device $m$ by $\boldsymbol{x}_m (t) \in \mathbb{C}^n$, where $\boldsymbol{x}_m (t) = \boldsymbol{0}_n$, if $m \notin \mathcal{M}(t)$. The channel gain from device $m$ to the PS is represented by
$h_m (t) \in \mathbb{C}$, which is i.i.d. according to $\mathcal{CN} (0, 1)$. The received signal at the PS includes an independent noise vector with each entry i.i.d. according to $\mathcal{CN} (0, \sigma^2)$. We assume that, at each iteration step, CSI for all the uplink channels is known by the PS, while each device knows its own CSI. The channel input of device $m$ at the $t$-th iteration is a function of the scheduling policy, channel gain $h_m(t)$, local dataset $\mathcal{B}_m$, and $\Delta \boldsymbol{\theta}_m (t)$, $m \in [M]$. For a total of $T$ global iterations, we impose the following average transmit power constraint on device $m$:
\begin{align}\label{AvePowerConsGen}
\frac{1}{T} \sum\nolimits_{t=0}^{T-1} \mathbb{E} \left[ ||\boldsymbol{x}^n_{m} (t)||^2_2 \right] \le \bar{P}, \quad \forall m \in [M],
\end{align}
where the expectation is over the randomness of the channel.

At each global iteration, the goal at the PS is to recover $\frac{1}{K} \sum\nolimits_{m \in \mathcal{M} (t)} \Delta \boldsymbol{\theta}_m (t)$ with the highest fidelity, which is used to update the global model parameters as in \eqref{PSGlobalUpdateModelDifScheduled}. The PS instead uses an estimate of $\frac{1}{K} \sum\nolimits_{m \in \mathcal{M} (t)} \Delta \boldsymbol{\theta}_m (t)$ upon receiving the noisy observation $\boldsymbol{y}(t)$ from the bandwidth-limited wireless medium to update the global model parameter vector, which is then shared among the devices through an error-free shared link for further computations.

We focus on digital transmissions from the devices to the PS, where each scheduled device employs data compression followed by channel coding to transmit its local model updates. We design various scheduling policies, and perform bandwidth allocation among the scheduled devices to have interference-free communication from the participating devices to the PS. In this digital approach, we assume capacity-achieving channel codes are employed by the devices, while the model updates are quantized at a resolution afforded by the channel capacity to guarantee their reliable transmission to the PS. 

At each model update stage, device $m$ is allocated $n_m$ distinct time slots, such that $\sum\nolimits_{m = 1}^{M} n_m = n$, where $n_m = 0$, if $m \notin \mathcal{M} (t)$. For large enough $n_m$, we assume that the Shannon capacity at device $m$ can be achieved; that is, the total amount of information that can be conveyed from device $m$ to the PS is given by $R_m (t) = n_m C_m (t)$, where $C_m (t) \triangleq \log_2 \left( 1 + \left| h_m(t) \right|^2 P_m (t) / \sigma^2 \right)$ and $P_m(t) \triangleq \mathbb{E} \left[ ||\boldsymbol{x}^n_{m} (t)||^2_2 \right]$, $m \in [M]$.

\section{Digital SGD (D-SGD) Quantization Scheme}\label{SecPeoposedDigital}


Here we present the data compression scheme employed by the devices for digital transmission over the wireless channel. 
We utilize the technique introduced in \cite{DCSattlerSparseBinary}, and extended in \cite{MohammadDenizDSGDCS} for FL over a bandwidth-limited wireless medium. 

It is worth noting that, at global iteration $t$, device $m$ intends to transmit $\Delta \boldsymbol{\theta}_m (t)$, computed after the $\tau$-step SGD algorithm, $m \in [M]$. For this purpose, it first quantizes $\Delta \boldsymbol{\theta}_m (t)$ by setting all but the largest $q_m(t)$ and the smallest $q_m(t)$ entries to zero (in practice, we typically have $q_m(t) \ll d$). It then computes the average of the positive and negative entries denoted by $q_m^+(t)$ and $q_m^-(t)$, respectively. If $q_m^+(t) \ge \left| q_m^-(t) \right|$, it sets all the negative entries to zero and all the positive entries to $q_m^+(t)$, and vice versa, if $q_m^+(t) < \left| q_m^-(t) \right|$. We denote the resultant quantized vector with $q_m(t)$ nonzero entries at device $m$ by $\Delta \widehat{\boldsymbol{\theta}}_m \left(q_m(t)\right)$, $m \in [M]$. To transmit $\Delta \widehat{\boldsymbol{\theta}}_m \left(q_m(t)\right)$, device $m$ requires 32 bits representing the real value $q_m^+(t)$ or $q_m^-(t)$ plus 1 bit for its sign, and no more than $\log_2 \binom{d}{q_m (t)}$ bits representing the locations of the nonzero entries, $m \in [M]$. Thus, device $m$ needs to transmit a total of
\begin{align}\label{DSGDNumberBits}
r_m \left(q_m (t)\right) = \log_2 \dbinom{d}{q_m (t)} + 33 \mbox{ bits}, \quad m \in [M].
\end{align}
The D-SGD quantization scheme at device $m$ is characterized by $q_m (t)$, and represented by ${\mbox{D-SGD} \left( q_m (t) \right)}$ resulting in $\Delta \widehat{\boldsymbol{\theta}}_m \left(q_m(t)\right)$, $m \in [M]$. The value of $q_m (t)$ is a design parameter that is determined for different scheduling policies, described in the next section, to satisfy the capacity limitation of transmission from device $m$ to the PS, $m \in [M]$.  

\section{Device Scheduling Policies}\label{SecScheduling}
In FL, the goal is typically to schedule as many devices as possible at each iteration \cite{KaibinParallelWork}; however, when the updates are conveyed to the PS over a shared wireless medium with limited bandwidth, having more devices scheduled at each iteration means that each device is allocated fewer resources, and as a result, contributes to the learning task with less accurate information. 
The goal is to identify the set of scheduled devices at each iteration that results in the best performance.

After receiving $\boldsymbol{\theta} (t)$ from the PS, all the devices perform the $\tau$-step SGD algorithm as in \eqref{ithStepSGDDevicem}. However, only $K \le M$ devices in $\mathcal{M} (t)$ are scheduled for transmission at iteration $t$, and the PS updates the global model parameter vector as follows:
\begin{align}\label{PSGlobalUpdateModelDifScheduledHat}
\boldsymbol{\theta} (t+1) = \boldsymbol{\theta} (t) + \frac{1}{K} \sum\nolimits_{m \in \mathcal{M} (t)} \Delta \widehat{\boldsymbol{\theta}}_m (q_m(t)). 
\end{align}

We take into account the channel conditions and the significance of the local model updates at the devices as the scheduling metrics. We study four different scheduling policies, namely the \textit{best channel} (BC), \textit{best $l_2$-norm} (BN2), \textit{best channel-best $l_2$-norm} (BC-BN2), and \textit{best $l_2$-norm-channel} (BN2-C) schemes, which we explain below. 

Due to the natural symmetry of the considered model across the devices, both in terms of the channel statistics and the local model updates, it is reasonable to assume that the probability of scheduling each device will be the same\footnote{We will indeed see below that this assumption holds for all four scheduling policies considered in this paper.}, given by $K/M$. 
Hence, the average transmit power constraint can be rewritten as follows:
\begin{align}\label{AvePowerConsGenUpdated}
\frac{K}{MT} \sum\nolimits_{t=1}^{T} P_m (t) \le \bar{P}, \quad \forall m \in [M].
\end{align}
For simplicity, we assume a fixed transmission power over time for the scheduled devices, $P_m (t) = M \bar{P}/K$, $\forall m \in \mathcal{M} (t)$, $\forall t$.

\subsection{BC Scheduling Policy}\label{SubSecBCSchePol}
The BC policy schedules devices based only on their channel gains. This generalizes the approach studied in \cite{FLTWCMohammadDenizFading}, where only a single device is scheduled based on the channel conditions. With BC, the PS does not require any information about the model updates at the devices, and it schedules $K$ devices with the highest channel gain magnitudes; that is, 
\begin{align}\label{BCK_t}
\mathcal{M} (t) = {\max}_{[K]} \left\{ \left| h_1 (t) \right|, ..., \left| h_M (t) \right|  \right\}.    
\end{align}
Having no knowledge about the model updates at the devices, the PS allocates the bandwidth so that the scheduled devices can each transmit the same number of bits. Given $\mathcal{M} (t) = \{ m_1, ..., m_K \}$, we want
\begin{align}\label{BCBandAllocEquat_1}
n_{m_1} C_{m_1}(t) = n_{m_2}& C_{m_2}(t) = \cdots = n_{m_K} C_{m_K}(t),
\end{align}
which, having $\sum\nolimits_{k=1}^{K} n_{m_k} = n$, results in 
\begin{align}\label{BCBandAllocEquatSubChann}
n_{m_k} = \frac{\prod_{i=1, i \ne k}^{K} C_{m_i}(t)}{\sum_{j=1}^{K} \prod_{i=1, i \ne j}^{K} C_{m_j}(t)}n, \quad k \in [K].     
\end{align}
After the above bandwidth allocation scheme, device $m_k$ performs quantization D-SGD$(q_{m_k} (t))$, with $q_{m_k} (t)$ set as the largest integer satisfying $r_m \left(q_{m_k} (t)\right) \le n_{m_k}C_{m_k}(t)$, $k \in [K]$, and transmits the quantized bits to the PS over the time slots allocated to it.

\subsection{BN2 Scheduling Policy}\label{SubSecBN2SchePol}
With BN2, the scheduling decision depends only on the significance of the model updates at the devices captured by the $l_2$-norm of the model update, $\left\| \Delta \boldsymbol{\theta}_m (t) \right\|_2$. Transmissions take place in two phases, where in the first phase, having computed $\Delta \boldsymbol{\theta}_m (t)$, device $m$ sends $\left\| \Delta \boldsymbol{\theta}_m (t) \right\|_2$ reliably to the PS, $\forall m \in [M]$. The PS then schedules $K$ devices with the largest $\left\| \Delta \boldsymbol{\theta}_m (t) \right\|_2$ values, i.e.,
\begin{align}\label{BN2K_t}
\mathcal{M} (t) = {\max}_{[K]} \left\{ \left\| \Delta \boldsymbol{\theta}_1 (t) \right\|_2, ..., \left\| \Delta \boldsymbol{\theta}_M (t) \right\|_2  \right\}.    
\end{align}
The bandwidth is allocated to the scheduled devices by the PS such that their link capacities are proportional to the significance of their local model updates; that is, for $m_i , m_{j} \in \mathcal{M} (t)$, $\forall i, j \in [K]$, $i \ne j$, we set
\begin{align}\label{BN2BandAllocEquat}
\frac{n_{m_i} C_{m_i} (t)}{n_{m_{j}} C_{m_{j}} (t)} = \frac{\left\| \Delta \boldsymbol{\theta}_{m_i} (t) \right\|_2}{\left\| \Delta \boldsymbol{\theta}_{m_{j}} (t) \right\|_2}.     
\end{align}
Having $\sum\nolimits_{k=1}^{K} n_{m_k} = n$, it follows that, for $k \in [K]$,
\begin{align}\label{BN2BandAllocEquatSubChann}
n_{m_k} = \frac{\prod_{i=1, i \ne k}^{K} C_{m_i}(t) \left\| \Delta \boldsymbol{\theta}_{m_i} (t) \right\|_2}{\sum_{j=1}^{K} \prod_{i=1, i \ne j}^{K} C_{m_j}(t) \left\| \Delta \boldsymbol{\theta}_{m_j} (t) \right\|_2}n.      
\end{align}
In the second phase of transmission, device $m_k$ transmits the result of D-SGD$(q_{m_k} (t))$, where $q_{m_k} (t)$ is set as the largest integer satisfying $r_{m_k} \left(q_{m_k} (t)\right) \le n_{m_k}C_{m_k}(t)$, $k \in [K]$.

\subsection{BC-BN2 Scheduling Policy}\label{SubSecBC_BN2SchePol}
BC-BN2 generalizes BC and BN2 by taking into account both the channel conditions and the significance of the model updates at the devices. The PS first identifies $K_c$ devices with the best channel conditions, for some $K \le K_c \le M$. Then, among these $K_c$ devices $K$ with the most significant model updates are scheduled. Formally, the PS first selects the best $K_c$ devices according to their channel states as follows:
\begin{align}\label{BC_BN2Kc_t}
\mathcal{M}_c (t) \triangleq {\max}_{[K_c]} \left\{ \left| h_1 (t) \right|, ..., \left| h_M (t) \right|  \right\}.    
\end{align}
Only these selected $K_c$ devices share $\left\| \Delta \boldsymbol{\theta}_m (t) \right\|_2$ with the PS, which schedules $K$ devices among them as follows:
\begin{align}\label{BC_BN2K_t}
\mathcal{M} (t) = {\max}_{[K]} \left\{ \left\| \Delta \boldsymbol{\theta}_m (t) \right\|_2, \forall m \in \mathcal{M}_c (t) \right\}.    
\end{align}
Having known $\left\| \Delta \boldsymbol{\theta}_{m_k} (t) \right\|_2, \forall m_k \in \mathcal{M} (t)$, we follow the same bandwidth allocation scheme as BN2.
Thus, device $m_k$, sends $\Delta \widehat{\boldsymbol{\theta}}_{m_k} \left(q_{m_k}(t)\right) = \mbox{D-SGD}\left(q_{m_k}(t)\right)$ to the PS, where $q_{m_k}(t)$ is set as the largest integer that satisfies $r_m \left(q_{m_k} (t)\right) \le n_{m_k}C_{m_k}(t)$, $k \in [K]$, where $n_{m_k}$ is given in \eqref{BN2BandAllocEquatSubChann}.   

We highlight that BC-BN2 for $K_c = K$ and $K_c = M$ corresponds to BC and BN2, respectively. 

\subsection{BN2-C Scheduling Policy}\label{SubSecBN2_CSchePol}
With BN2-C, each device performs the D-SGD quantization scheme assuming the availability of all the bandwidth to itself, and finds the resultant quantized vector, whose $l_2$-norm is sent to the PS, based on which it schedules the devices. To be more precise, device $m$ calculates the D-SGD quantization parameter, denoted by $q_m^* (t)$, satisfying $r_m \left(q^*_{m} (t)\right) \le C_{m}(t)$, $\forall m \in [M]$, based on which it computes $\Delta \widehat{\boldsymbol{\theta}}_m \left(q^*_m(t)\right) = \mbox{D-SGD}\left(q^*_{m}(t)\right)$; it then shares $|| \Delta \widehat{\boldsymbol{\theta}}_m \left(q^*_m(t)\right) ||_2$ with the PS. The PS schedules the devices with the best estimated $l_2$-norm after quantization:
{\small{\begin{align}\label{BN2CK_t}
\mathcal{M} (t) = {\max}_{[K]} \left\{ || \Delta \widehat{\boldsymbol{\theta}}_1 \left(q^*_1(t)\right) ||_2, ..., || \Delta \widehat{\boldsymbol{\theta}}_M \left(q^*_M(t)\right) ||_2 \right\},    
\end{align}}}
\hspace{-.16cm}and, given $m_i, m_j \in \mathcal{M}(t)$, allocates the bandwidth to the scheduled devices such that
\begin{align}\label{BN2_CBandAllocEquat}
\frac{n_{m_i} C_{m_i} (t)}{n_{m_{j}} C_{m_{j}} (t)} = \frac{|| \Delta \widehat{\boldsymbol{\theta}}_{m_i} (q^*_{m_i}(t)) ||_2}{|| \Delta \widehat{\boldsymbol{\theta}}_{m_j} (q^*_{m_j}(t)) ||_2}, \quad \forall i, j \in [K], i \ne j.     
\end{align}
With $\sum\nolimits_{k=1}^{K} n_{m_k} = n$, it follows that, for $k \in [K]$, 
\begin{align}\label{BN2_CBandAllocEquatSubChann}
n_{m_k} = \frac{\prod_{i=1, i \ne k}^{K} C_{m_i}(t) || \Delta \widehat{\boldsymbol{\theta}}_{m_i} (q^*_{m_i}(t)) ||_2}{\sum_{j=1}^{K} \prod_{i=1, i \ne j}^{K} C_{m_j}(t) || \Delta \widehat{\boldsymbol{\theta}}_{m_j} (q^*_{m_j}(t)) ||_2}n. 
\end{align}
Scheduled device $m$ performs the D-SGD$(q_{m}(t))$ quantization scheme, where $q_{m}(t)$ is set as the largest integer satisfying $r_m \left(q_{m} (t)\right) \le n_{m}C_{m}(t)$, with $n_m$ given in \eqref{BN2_CBandAllocEquatSubChann}, $\forall m \in \mathcal{M}(t)$. 

\begin{remark}
We highlight that BN2-C intertwines the channel conditions and the significance of local model updates to schedule the devices. 
Unlike BN2 and BC-BN2, where $\left\| \Delta \boldsymbol{\theta}_m (t) \right\|_2$ is directly used for scheduling, BN2-C utilizes the output of the D-SGD quantization scheme, $|| \Delta \widehat{\boldsymbol{\theta}}_m \left(q^*_m(t)\right) ||_2$, for scheduling, where $q^*_m(t)$ is a function of the channel gain.  
This novel technique comes with a computational cost at the devices due to the extra computation of $\Delta \widehat{\boldsymbol{\theta}}_m \left(q^*_m(t)\right)$. On the other hand, BC requires the smallest computational overhead at the devices. In this work, we do not study the computational complexity at the devices, and the main goal is to utilize the limited communication resources efficiently.  
\end{remark}

\section{Convergence Analysis}\label{SecConvergence}
Here we provide a convergence analysis for the FL over wireless networks with device scheduling, where the participating devices cannot transmit their local model updates in entirety due to the resource limitations and unreliable wireless links. For the convergence proof, we consider a slightly different quantization technique than the D-SGD quantization scheme, and we show that the gap between the expected loss function and the minimum loss function approaches zero for large enough $T$. 
We first present the preliminaries and assumptions, and then the convergence result, whose proof is provided in the Appendix.

\subsection{Preliminaries}\label{SubSecPreConvergence}
We define the optimal solution as 
\begin{align}\label{ThetaStarDef}
\boldsymbol{\theta}^* \triangleq \arg \mathop {\min }\limits_{\boldsymbol{\theta}} F(\boldsymbol{\theta}),  
\end{align}
and the minimum loss as $F^* \triangleq F(\boldsymbol{\theta}^*)$. We also denote the minimum value of $F_m (\cdot)$, the loss function at device $m$, by $F_m^*$, $m \in [M]$. We then define
\begin{align}\label{GammaDef}
\Gamma \triangleq F^* - \frac{1}{M} \sum\limits_{m=1}^{M} F^*_m,    
\end{align}
where $\Gamma \ge 0$, and its magnitude indicates the bias in the data distribution across devices. For i.i.d. data distribution, given enough number of local data samples, $B$, $\Gamma$ approaches zero.

For the simplicity of the convergence analysis, we consider $\eta_m^i (t) = \eta(t)$, $\forall m ,i$. Thus, the $i$-th step local SGD at device $m$ is given by
\begin{align}\label{ConvSGDDevicem}
\boldsymbol{\theta}_m^{i+1} (t) = \boldsymbol{\theta}_m^i (t) - \eta (t) \nabla F_m \left( \boldsymbol{\theta}_m^i (t), \xi_m^i (t) \right),  \quad \mbox{$i \in [\tau]$}, \mbox{$m \in [M]$},   
\end{align}
where $\mathbb{E}_{\xi} \left[ \nabla F_m \left( \boldsymbol{\theta} (t), \xi_m (t) \right) \right] = \nabla F_m \left( \boldsymbol{\theta} (t)  \right)$, and $\boldsymbol{\theta}_m^1 (t) = \boldsymbol{\theta} (t)$. Thus, we have
\begin{align}\label{ConvSGDDevicemith_2}
\boldsymbol{\theta}_m^{\tau+1} (t) = \boldsymbol{\theta} (t) - \eta (t) \sum\nolimits_{i=1}^{\tau} \nabla F_m \left( \boldsymbol{\theta}_m^i (t), \xi_m^i (t) \right).    
\end{align}
We define 
\begin{align}\label{ConvgmtDefinition}
\boldsymbol{g}_m (t) \triangleq - \eta (t) \sum\nolimits_{i=1}^{\tau} \nabla F_m \left( \boldsymbol{\theta}_m^i (t), \xi_m^i (t) \right), \quad m \in [M].     
\end{align}
Device $m$ transmits a sparsified version of the local model update $\Delta \boldsymbol{\theta}_m (t) = \boldsymbol{g}_m (t)$, denoted by $\widehat{\boldsymbol{g}}_m (t)$, $m \in [M]$. Upon receiving $\widehat{\boldsymbol{g}}_m (t)$ from $K$ devices in $\mathcal{M} (t)$, the PS updates the model as follows:
\begin{align}\label{ConvPSGlobalUpdateModelDifScheduled}
\boldsymbol{\theta} (t+1) = \boldsymbol{\theta} (t) + \frac{1}{K} \sum\limits_{m \in \mathcal{M} (t)} \widehat{\boldsymbol{g}}_m (t).
\end{align}

To simplify the analysis, we assume that each participating device has the same link capacity. This can be guaranteed through the resource allocation technique given in \eqref{BCBandAllocEquatSubChann}. For further simplification, we consider the following sparsification scheme. Each device sparsifies its local model update
in a random fashion;
that is, at iteration $t$, it sets all but $q(t)$ entries, 
selected uniformly at random,
to zero. Then, it quantizes the remaining entries using a $32$-bit uniform quantization scheme. It is easy to verify that each device needs to transmit
\begin{align}\label{ConvDSGDNumberBits}
r \left(q (t)\right) = \log_2 \dbinom{d}{q (t)} + 33 \cdot q(t) \mbox{ bits},
\end{align}  
where $q (t)$ is set as the largest integer satisfying $r \left(q (t)\right) \le n_{m}C_{m}(t)$, for some $m \in \mathcal{M} (t)$, with $n_{m}$ given in \eqref{BCBandAllocEquatSubChann}.

We denote the set of all possible sparsity patterns of $\widehat{\boldsymbol{g}}_{m} (t)$ by $\cal S$, where $\left| \mathcal{S} \right| = \binom{d}{q(t)}$, $\forall m$, and define  $\rho(t) \triangleq \frac{q(t)}{d}$. We note that, once $K$ and the quantization scheme are fixed, $\rho(t)$ is an i.i.d. random variable whose distribution depends on the channel distribution.  


\begin{lemma}\label{LemmaMeanVarhatg}
Due to random sparsification, the sparsity pattern of each model update follows a uniform distribution. Thus, we have
\begin{subequations}\label{MeanVarhatg}
\begin{align}\label{Meanhatg}
\mathbb{E}_{S} \left[ \widehat{\boldsymbol{g}}_{m} (t) \right] = & \rho(t) {\boldsymbol{g}}_{m} (t), \\
\mathbb{E}_{S} \left[ \left\| \widehat{\boldsymbol{g}}_{m} (t) \right\|_2^2 \right] = & \rho(t) \left\| {\boldsymbol{g}}_{m} (t) \right\|_2^2, \quad \forall m \in [M], \label{Varhatg}
\end{align}
\end{subequations}
where $\mathbb{E}_S$ denotes the expectation with respect to the randomness of the sparsification technique. 
\end{lemma}
\begin{proof}
For any specific sparsity pattern $s \in \cal S$, we denote the sparse vector $\widehat{\boldsymbol{g}}_{m} (t)$ by ${\boldsymbol{g}}^s_{m} (t)$, $\forall m$. We have
\newcommand\zeroequal{\mathrel{\overset{\makebox[0pt]{\mbox{\normalfont\tiny\sffamily (a)}}}{=}}}
\begin{subequations}
\begin{align}
\mathbb{E}_{S} \left[ \widehat{\boldsymbol{g}}_{m} (t) \right] = & \sum\limits_{s \in \mathcal{S}} {\boldsymbol{g}}^s_{m} (t) \Pr \{ s \in \mathcal{S} \} = \frac{1}{\binom{d}{q(t)}} \sum\limits_{s \in \mathcal{S}} {\boldsymbol{g}}^s_{m} (t) \zeroequal \frac{\binom{d-1}{q(t)-1}}{\binom{d}{q(t)}} {\boldsymbol{g}}_{m} (t)  = \rho(t) {\boldsymbol{g}}_{m} (t),\\
\mathbb{E}_{S} \left[ \left\| \widehat{\boldsymbol{g}}_{m} (t) \right\|_2^2 \right] = &\sum\limits_{s \in \mathcal{S}} \left\| {\boldsymbol{g}}^s_{m} (t) \right\|_2^2 \Pr \{ s \in \mathcal{S} \} = \frac{1}{\binom{d}{q(t)}} \sum\limits_{s \in \mathcal{S}} \left\| {\boldsymbol{g}}^s_{m} (t) \right\|_2^2 \zeroequal \frac{\binom{d-1}{q(t)-1}}{\binom{d}{q(t)}} \left\| {\boldsymbol{g}}_{m} (t) \right\|_2^2 \nonumber\\
= & \rho(t) \left\| {\boldsymbol{g}}_{m} (t) \right\|_2^2,
\end{align}
\end{subequations}
where (a) follows due to the symmetry in both $\sum\limits_{s \in \mathcal{S}} {\boldsymbol{g}}^s_{m} (t)$ and $\sum\limits_{s \in \mathcal{S}} \left\| {\boldsymbol{g}}^s_{m} (t) \right\|_2^2$, where each entry of ${\boldsymbol{g}}_{m} (t)$ appears exactly $\binom{d-1}{q(t)-1}$ times.   
\end{proof}

\begin{assumption}\label{AssumpSmoothLoss}
Loss functions $F_1, \dots, F_M$ are all $L$-smooth; that is, $\forall \boldsymbol{v}, \boldsymbol{w} \in \mathbb{R}^d$, 
\begin{align}\label{ConvLSmoothCondit}
F_m(\boldsymbol{v}) - F_m(\boldsymbol{w}) \le \langle \boldsymbol{v} - \boldsymbol{w} , \nabla F_m (\boldsymbol{w}) \rangle + \frac{L}{2} \left\| \boldsymbol{v} - \boldsymbol{w} \right\|^2_2, \quad \forall m \in [M].
\end{align}
\end{assumption}

\begin{assumption}\label{AssumpStrongConvexLoss}
Loss functions $F_1, \dots, F_M$ are all $\mu$-strongly convex; that is, $\forall \boldsymbol{v}, \boldsymbol{w} \in \mathbb{R}^d$, 
\begin{align}\label{ConvMuStConvexCondit}
F_m(\boldsymbol{v}) - F_m(\boldsymbol{w}) \ge \langle \boldsymbol{v} - \boldsymbol{w} , \nabla F_m (\boldsymbol{w}) \rangle + \frac{\mu}{2} \left\| \boldsymbol{v} - \boldsymbol{w} \right\|^2_2, \quad \forall m \in [M].      
\end{align}
\end{assumption}

\begin{assumption}\label{AssumpBoundedVarGradient}
The expected squared $l_2$-norm of the stochastic gradients is bounded; that is,
\begin{align}\label{ConvNorm2Bound}
\mathbb{E}_{\xi} \left [ \left\| \nabla F_m \left( \boldsymbol{\theta}_m (t), \xi_m (t) \right) \right\|^2_2 \right] \le G^2, \quad \forall m \in [M], \; \forall t.      
\end{align}
\end{assumption}

\subsection{Convergence Result}\label{SubSecConvRes}
We assume that the device scheduling is uniformly random across the devices, which is consistent with the device scheduling policies introduced in Section \ref{SecScheduling}. Thus, the probability that a device is scheduled for transmission at any iteration is $\frac{K}{M}$.

\begin{theorem}\label{Theoremtheta_thetastarKequalM}
Let $0 < \eta(t) \le \min \left\{ 1, \frac{1}{\mu \tau} \right\}$, $\forall t$. We have
\begin{subequations}\label{ConvTheoremtheta_thetastarKequalM}
\begin{align}\label{ConvTheoremtheta_thetastarKequalM_main}
\mathbb{E} \left[ \left\| \boldsymbol{\theta} (t) - {\boldsymbol{\theta}}^* \right\|_2^2 \right] \le  \Bigg( \prod_{i=0}^{t-1} A(i) \Bigg) \left\| {\boldsymbol{\theta}} (0) - {\boldsymbol{\theta}}^* \right\|_2^2 + \sum_{j=0}^{t-1} B(j) \prod_{i=j+1}^{t-1} A(i),  
\end{align}
where 
\begin{align}\label{ConvTheoremtheta_thetastarKequalM_AB}
A(i) \triangleq & 1 - \mu \rho(i) \eta (i) \left( \tau - \eta(i) (\tau - 1) \right),\\ 
B(i) \triangleq & \frac{(M-K)\rho(i) \eta^2(i) \tau^2 G^2}{K(M-1)} + \rho(i) \left( 1+ \mu (1- \eta(i)) \right) \eta^2(i) G^2 \frac{\tau (\tau-1)(2\tau-1)}{6} \nonumber\\
&+ \rho(i) \eta^2(i) (\tau^2 + \tau-1) G^2  + 2 \rho(i) \eta(i) (\tau - 1) \Gamma, 
\end{align}
\end{subequations}
and the expectation is with respect to the stochastic gradient function, the quantization technique, and the randomness of device scheduling.
\end{theorem}
\begin{proof}
See Appendix \ref{AppProofTheoremTheta_ThetaStarKequalM}. 
\end{proof}

\begin{corollary}\label{CorrConvF_FstarKequalM}
From the $L$-smoothness of function $F(\cdot)$, after $T$ global iterations, we have
\begin{align}\label{ConvF_FstarKequalM}
\mathbb{E} \left[ F( \boldsymbol{\theta} (T)) \right] - F^* \le & \frac{L}{2} \mathbb{E} \left[ \left\| \boldsymbol{\theta} (T) - {\boldsymbol{\theta}}^* \right\|_2^2 \right] \nonumber \\
\le & \frac{L}{2} \prod_{i=0}^{T-1} A(i) \left\| {\boldsymbol{\theta}} (0) - {\boldsymbol{\theta}}^* \right\|_2^2 + \frac{L}{2} \sum_{j=0}^{T-1} B(j) \prod_{i=j+1}^{T-1} A(i), 
\end{align}
where the last inequality follows from Theorem \ref{Theoremtheta_thetastarKequalM}. Having a decreasing learning rate $\mathop {\lim }\limits_{t \to \infty } \eta(t) = 0$, it is easy to verify that $\mathop {\lim }\limits_{T \to \infty } \mathbb{E} \left[ F( \boldsymbol{\theta} (T)) \right] - F^* = 0$.
\end{corollary}

\begin{remark}\label{RemConvKequalMpi}
For a limited $T$ value, it is hard to clearly discuss the impact of $\rho(i)$ on the final average loss, since it depends on the other parameters, $\mu$, $L$, $M$, $K$, $\Gamma$ and $G$. However, since $A(i)$ reduces with $\rho(i)$, one can observe that the convergence speed increases with $\rho(i)$. 
\end{remark}

\begin{remark}\label{RemConvKlessMpi}
The first term in $B(i)$, which is due to device scheduling, is a decreasing function of $K$, where we note that $\rho(i)$ reduces with $K$ due to resource sharing. Thus, $B(i)$ reduces with $K$, and $K = M$ minimizes $B(i)$. On the other hand, since $A(i)$ increases with $K$, which is due to the reduction in $\rho(i)$, the impact of $K$ on the convergence performance is complicated. We will observe in Fig. \ref{Fig_Convergence} that for different setting parameters different $K$ values provide the best convergence performance.     
However, when the resources are abundant such that $\rho(t) = 1$, $\forall t$, it is trivially known that the full devices participation scenario, i.e., $K=M$, provides the best performance; this is corroborated by the result in Corollary \ref{CorrConvF_FstarKequalM}, where, for $\rho(t) = 1$, $\forall t$, $K=M$ gives the best performance.
\end{remark}

\begin{remark}\label{RemTau}
Similarly to $\rho(i)$, the impact of $\tau$ on the convergence rate is complicated. From $A(i)$, it is clear that the convergence speed increases with $\tau$, which verifies the observations made in \cite{StichLocalSGD}, where increasing the number of local model updates is proposed to increase the convergence speed. On the other hand, $B(i)$ increases with $\tau$, limiting the value of $\tau$ providing the best performance. We highlight that the last term in $B(i)$ reflects the impact of bias in the data distribution on the convergence rate. It is obvious that a higher $\tau$ value would intensify the deterioration of this term on the convergence rate.        
\end{remark}

\section{Numerical Experiments}\label{SecExperiments}

Here we compare the performance of different scheduling policies for image classification on the MNIST dataset \cite{LeCunMNIST} with $60000$ training and $10000$ test samples. We train a multi-layer perceptron neural network with a single hidden layer with $256$ parameters, in which case the total number of parameters is $d = 203530$, where \textit{softmax} is utilized as the activation function of the output layer.

\begin{figure}[t!]
\centering
\centering
\includegraphics[scale=0.8,trim={20pt 7pt 45pt 40pt},clip]{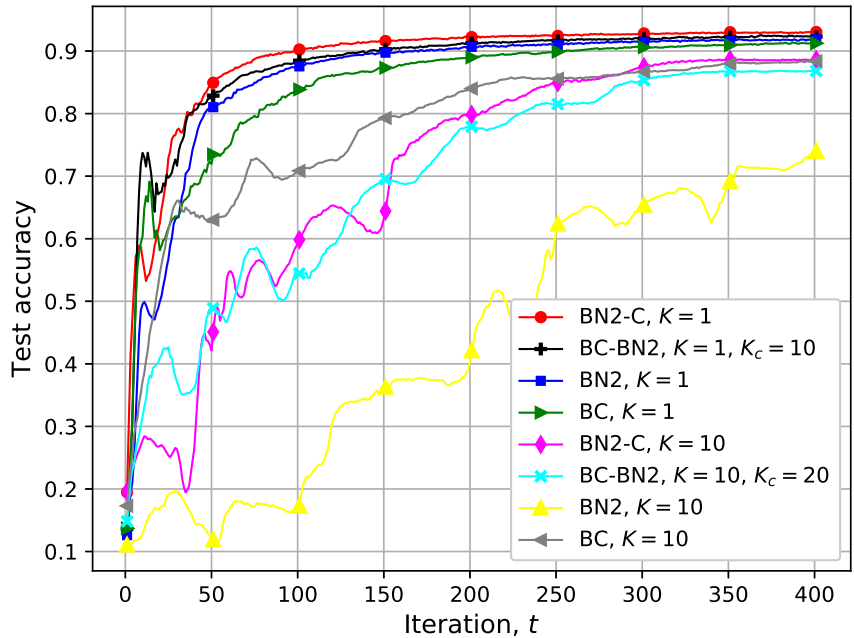}
\caption{Performance of different scheduling policies for IID data distribution with $M=40$, $B=1000$ and $n = 5 \times 10^3$.}
\label{Fig_IID_K_1_5}
\end{figure}

We consider two data distribution scenarios: in the \textit{IID} case, data samples at each device are selected at random from the training set; while in the \textit{non-IID}  case, each device has samples from only two labels/classes chosen randomly, and half of the local data samples are selected at random from each chosen label/class. We utilize ADAM \cite{ADAMDC} and AdaGrad \cite{AdaGradRef} optimizers to train the neural network for the IID and non-IID data distribution scenarios, respectively.

For the experiments, we consider $M=40$ devices, each with $B=1000$ training data samples. For the transmission of the local updates, we assume $n=5 \times 10^3$ symbols available at each global iteration, and we assume a noise variance of $\sigma^2 = 1$, and average power constraint $\bar{P} = 1$. We set the number of local iterations at the devices to $\tau = 3$. We measure the performance as the accuracy with respect to the test samples, called \textit{test accuracy}, versus the iteration count at the PS, $t$. For numerical experiments, we consider the proposed D-SGD quantization scheme since it leads to a better accuracy than the random sparsification technique introduced for the convergence analysis.   

In Fig. \ref{Fig_IID_K_1_5}, we compare the performance of different scheduling policies for the IID data distribution scenario. 
The goal here is to find the value of $K$ resulting in the best performance for each scheduling policy. To this end, we consider two different values, $K = 1$ and $K=10$, for each scheduling policy, where for BC-BN2 we set $K_c = 10$ and $K_c = 20$, respectively. We observe in Fig. \ref{Fig_IID_K_1_5} that, for each scheduling policy, increasing $K$ deteriorates the accuracy in terms of both the convergence speed and the final accuracy level. We did not include the results for other $K$ values, as we have observed that the performance of each scheduling policy deteriorates with increasing $K$. Thus, we focus on $K=1$, which, based on our observations, provides the best performance for each policy. 
This illustrates that, with IID local data samples, sending a more accurate update from a single device (which is scheduled at random thanks to the symmetry across the devices in our model) provides a faster convergence rate in the long-term than sending less accurate updates from multiple devices. 
For comparison, we provide the final accuracy level of each scheduling policy for $K=1$. These are given by $91.2\%$, $91.7\%$, $92.3\%$ and $93.1\%$ for BC, BN2, BC-BN2 and BN2-C, respectively.
As can be seen, BN2-C provides the best performance in terms of the convergence speed as well as the final accuracy level. The improvement of BC-BN2 over BN2 is marginal, but both outperform BC. These results illustrate that, given IID data distribution, scheduling devices according to both the significance of their model updates and their channel conditions provides gains in terms of accuracy. Also, from the superiority of BN2 over BC, we conclude that, to obtain the best performance for the IID scenario, the significance of the model updates, captured by the $l_2$-norm of the local model updates, plays a more important role in the accuracy performance than the channel conditions.
On the other hand, for large $K$, such as $K=10$ (which does not provide the best performance in this experiment), it is important to consider the channel conditions for scheduling in order to make sure that the scheduled devices can send sufficient information about their local updates, rather than scheduling the devices based only on the $l_2$-norm of their model updates.   

\begin{figure}[t!]
\centering
\begin{subfigure}{.5\textwidth}
  \centering
  \includegraphics[scale=0.55,trim={20pt 7pt 36pt 40pt},clip]{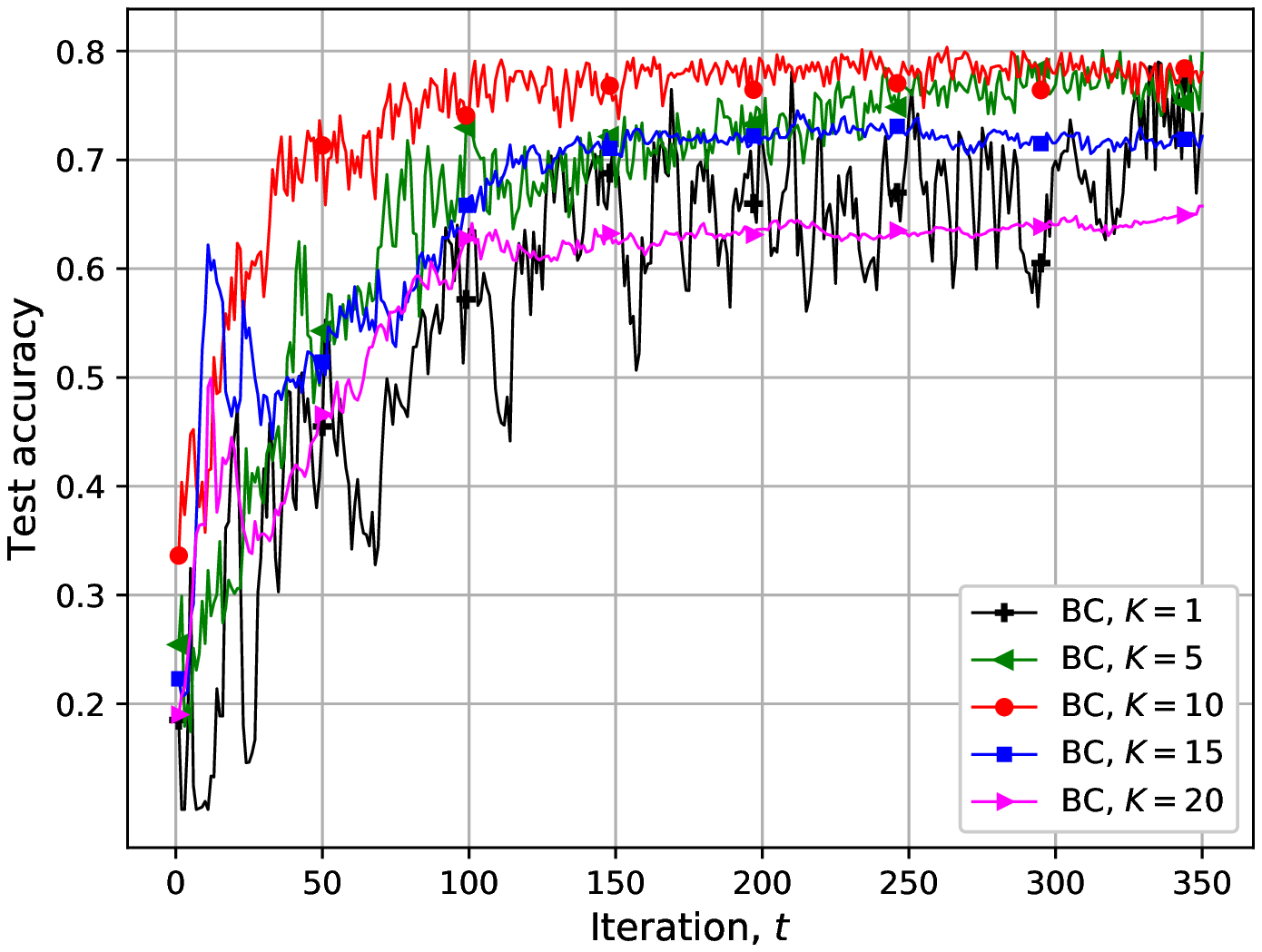}\vspace{0cm}
  \caption{BC scheduling policy}
  \label{Fig_nonIID2_BC}
\end{subfigure}%
\begin{subfigure}{.5\textwidth}
  \centering
  \includegraphics[scale=0.55,trim={20pt 7pt 36pt 40pt},clip]{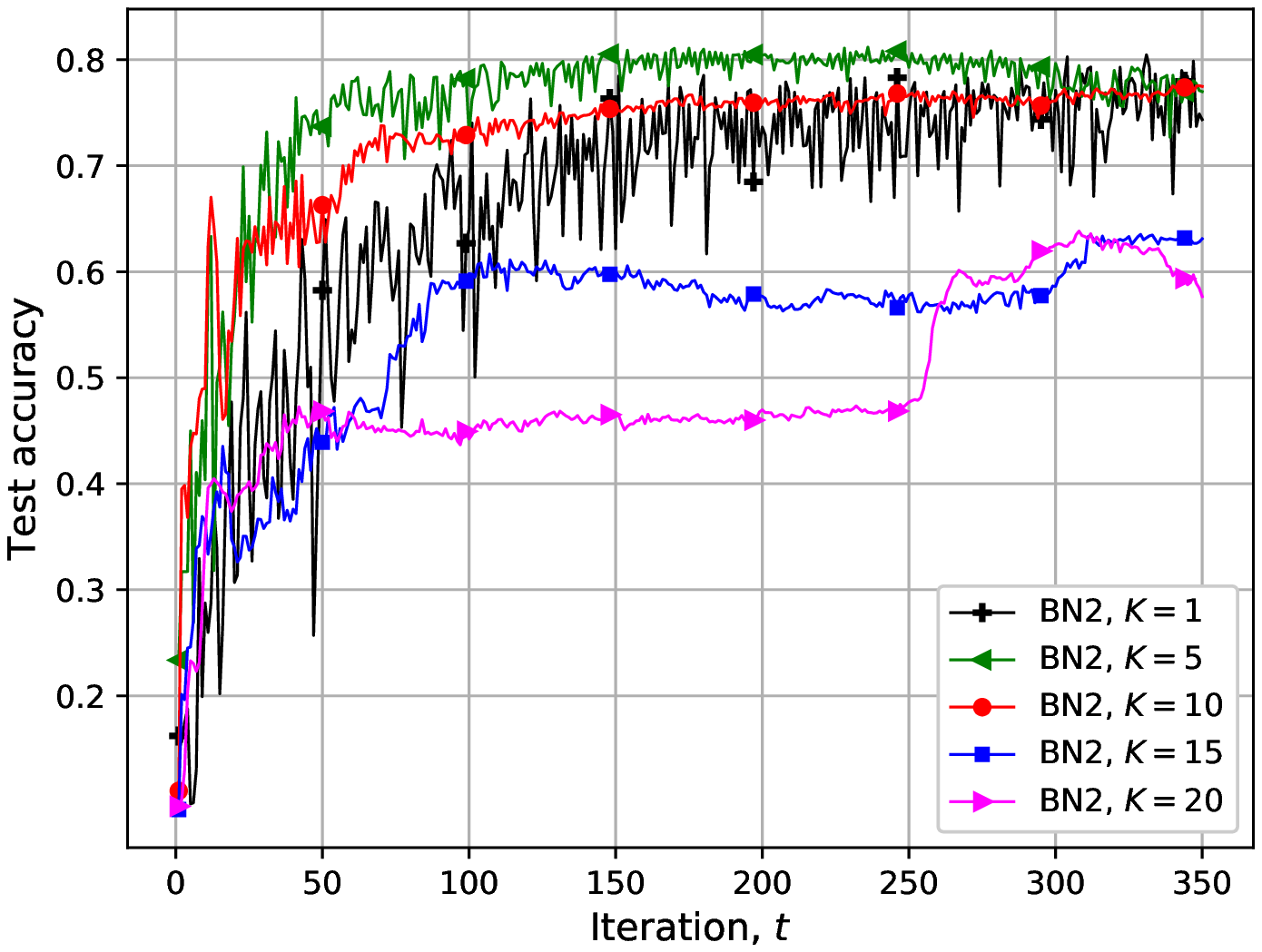}\vspace{0cm}
  \caption{BN2 scheduling policy}
  \label{Fig_nonIID2_BN2}
\end{subfigure}\\\vspace{.2cm}
\begin{subfigure}{.49\textwidth}
  \centering
  \includegraphics[scale=0.55,trim={20pt 7pt 36pt 40pt},clip]{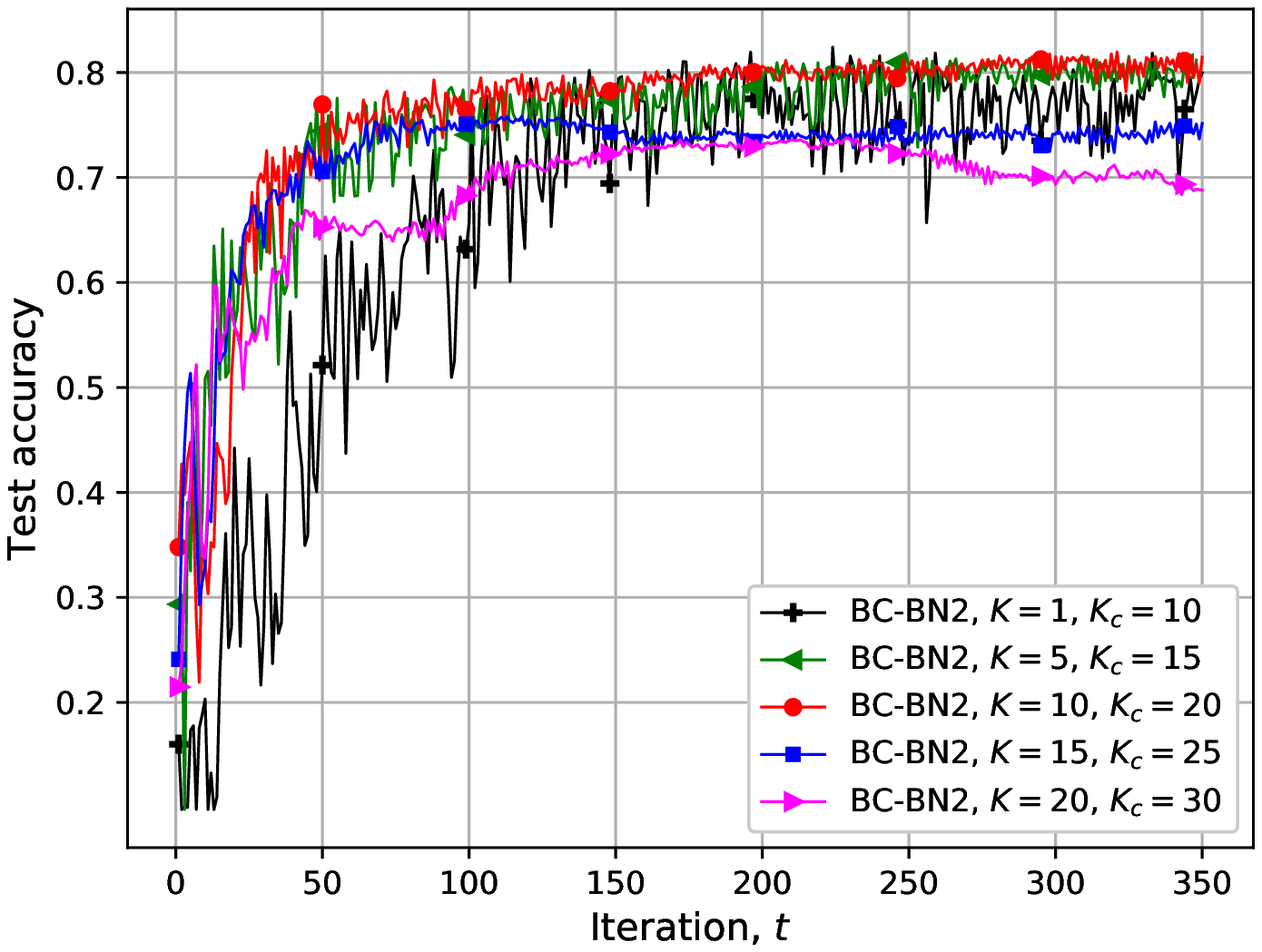}\vspace{0cm}
  \caption{BC-BN2 scheduling policy}
  \label{Fig_nonIID2_BC_BN2}
\end{subfigure}
\begin{subfigure}{.49\textwidth}
  \centering
  \includegraphics[scale=0.55,trim={20pt 7pt 36pt 40pt},clip]{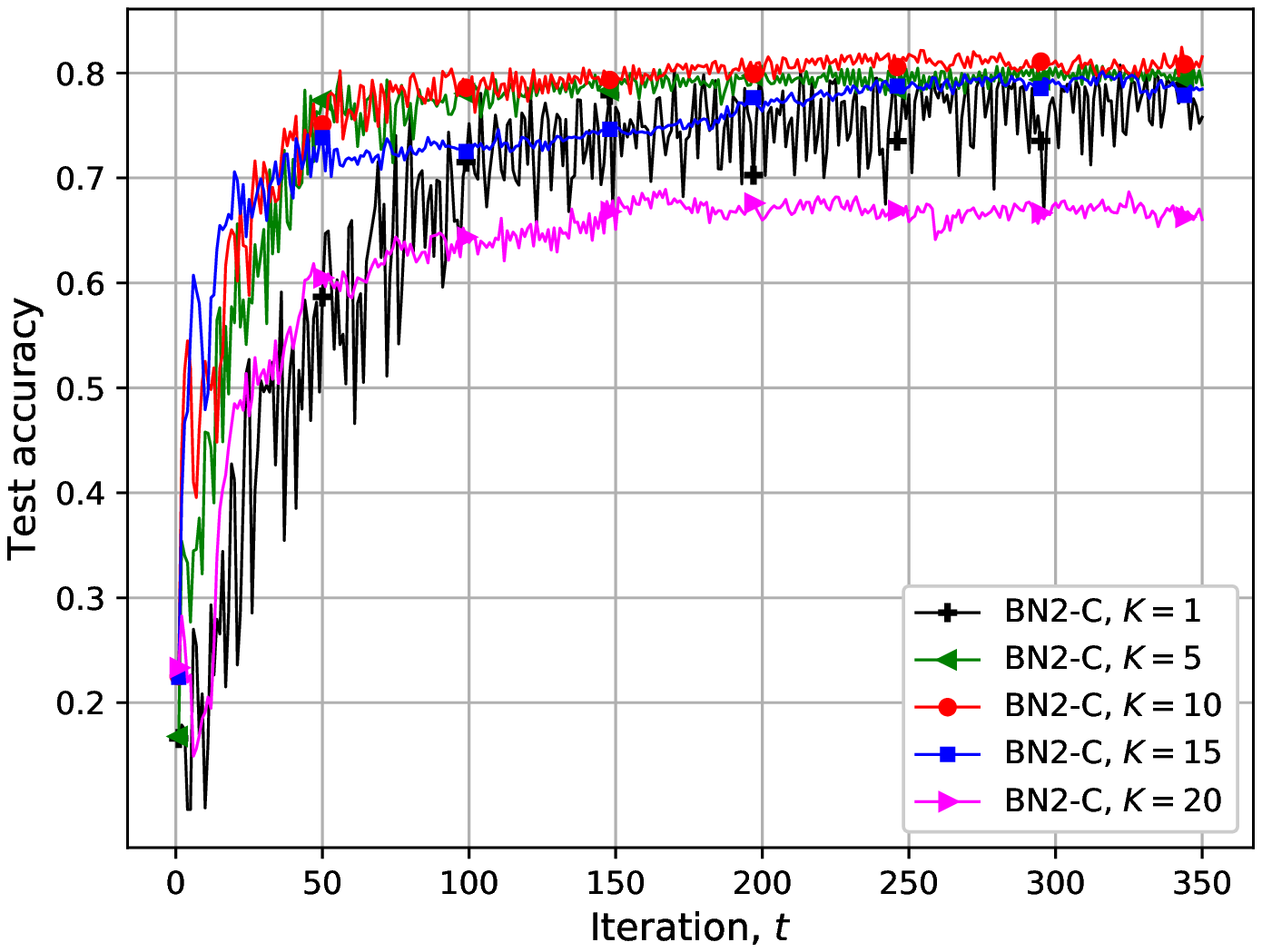}\vspace{0cm}
  \caption{BN2-C scheduling policy}
  \label{Fig_nonIID2_BN2_C}
\end{subfigure}
\caption{Performance of different scheduling policies for non-IID data distribution with $M=40$, $B=1000$ and $n = 5 \times 10^3$.}
\label{Fig_nonIID2}
\end{figure}

In Fig. \ref{Fig_nonIID2}, we investigate the performance of these different scheduling policies for the non-IID scenario. As can be seen, for all the scheduling policies, unlike in the IID case, scheduling a single device results in instability of the learning performance appearing as fluctuations in their accuracy levels over iterations. In the non-IID scenario, the local model update at each device is biased due to the biased local datasets, and scheduling a single device provides inaccurate information and causes instability in the performance in the long-term. On the other hand, increasing $K$ (sharing resources among more devices) reduces the accuracy at which the scheduled devices can transmit their model updates. As a result, it is expected that a moderate $K$ value would provide the best performance, which is confirmed with our simulation results. For the setting under consideration, $K=10$ provides the best final accuracy for BC, BC-BN2 and BN2-C, while $K=5$ performs better for BN2, although $K=10$ shows a more stable accuracy performance with a higher final accuracy level. Similarly to the IID scenario, we observe that it is essential to consider the channel conditions for higher $K$ values in order to make sure that the devices can transmit enough information. Also, as can be seen from the performance of BC for $K=10$, when scheduling based only on the channel conditions, the performance is more unstable, unless a relatively large number of devices are scheduled, in which case the accuracy level deteriorates. 
We highlight that, compared to the channel conditions, scheduling based on the significance of the model updates has a greater impact on the performance at the initial iterations when the gradients are more aggressive. On the other hand, it is important to consider the channel conditions at later iterations when approaching the optimum solution, since the SGD algorithm is more vulnerable to the noise compared to the initial iterations, and a more accurate estimate of the model update at each participating device is required for robust communication against the noise.
For comparison, the best final accuracy levels for BC, BN2, BC-BN2 and BN2-C are $78\%$, $77.5\%$, $81.5\%$ and $81.7\%$, respectively.
It can be seen that BN2-C and BC-BN2 outperform BC and BN2 in terms of the accuracy level, highlighting the importance of scheduling devices based on both the channel conditions and the model updates at the devices for the non-IID scenario.

\begin{figure}[t!]
\centering
\begin{subfigure}{.5\textwidth}
  \centering
  \includegraphics[scale=0.55,trim={8pt 7pt 36pt 40pt},clip]{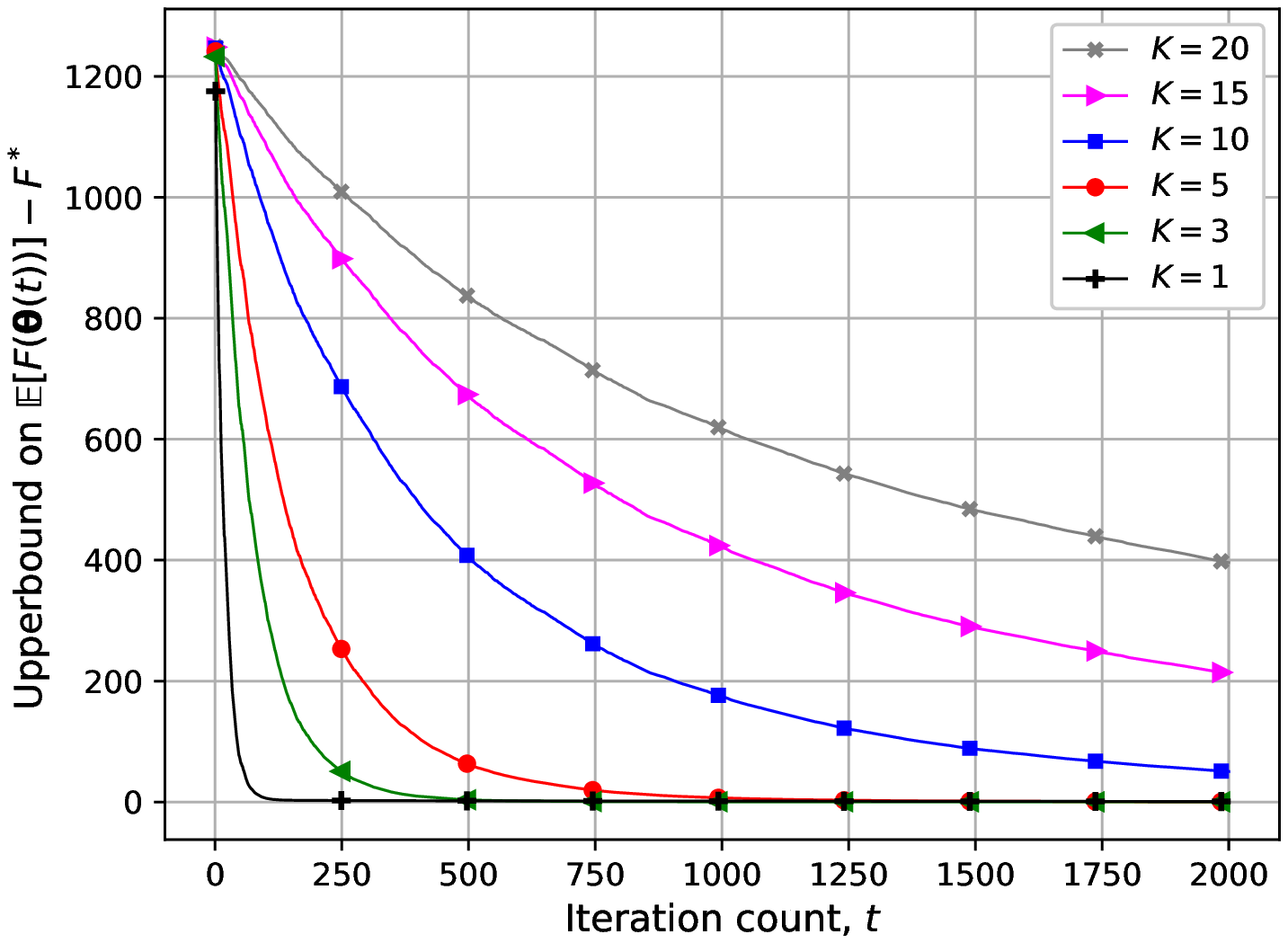}\vspace{0cm}
  \caption{$\Gamma = 1$ and $G=1$}
  \label{Fig_Convergence_G1}
\end{subfigure}%
\begin{subfigure}{.5\textwidth}
  \centering
  \includegraphics[scale=0.55,trim={8pt 7pt 36pt 40pt},clip]{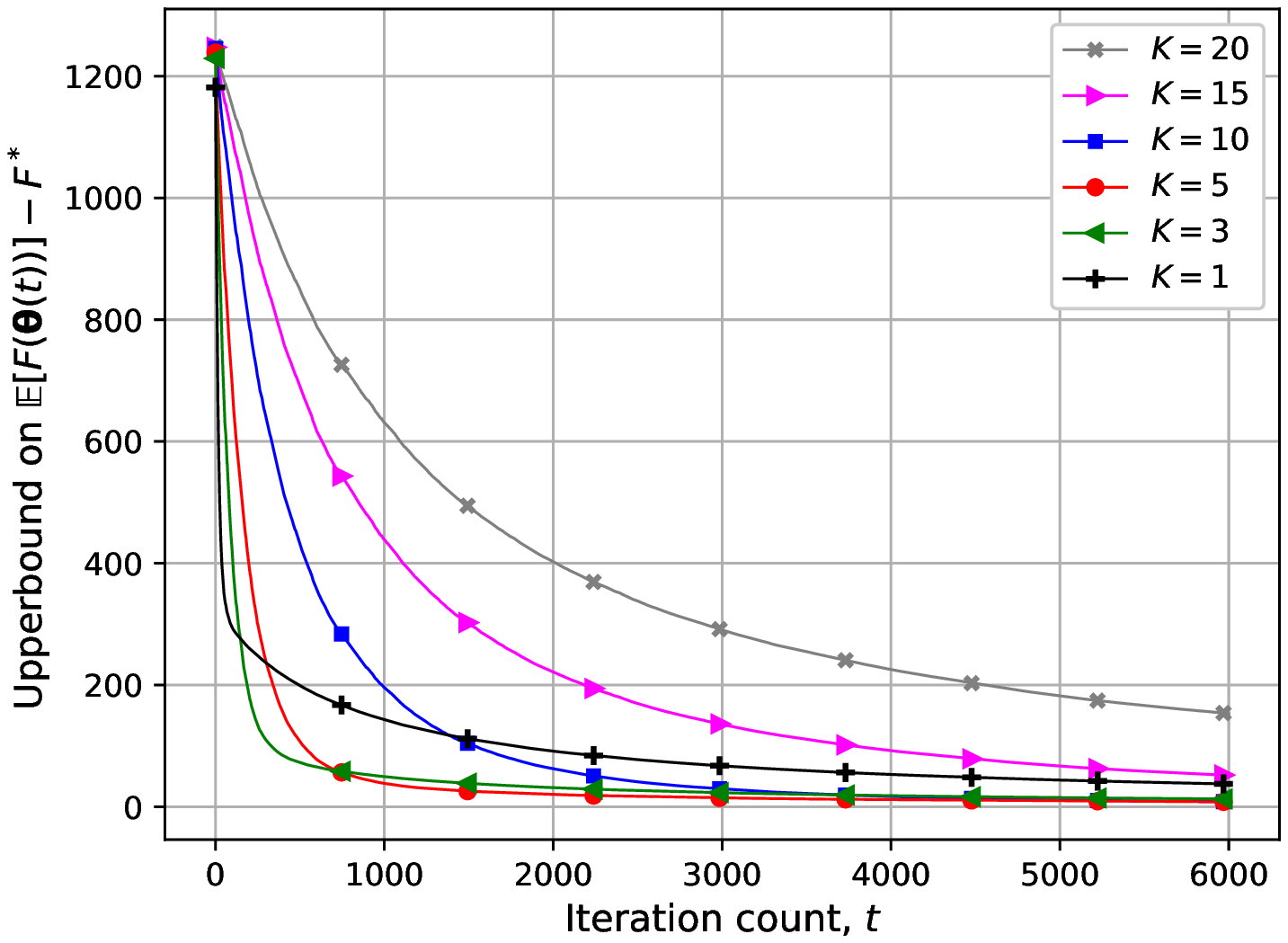}\vspace{0cm}
  \caption{$\Gamma = 10$ and $G=10$}
  \label{Fig_Convergence_G100}
\end{subfigure}
\caption{Upper bound on $\mathbb{E} \left[ F( \boldsymbol{\theta} (T)) \right] - F^*$ for different $K$ values, $K \in \{ 1, 3, 5, 10, 15, 20\}$, considering $M=100$, $n = 10^5$, $\tau = 3$, $\mu = 1$, $L=5$, $\left\| \boldsymbol{\theta}(0) - \boldsymbol{\theta}^* \right\|_2^2 = 500$ and $\eta (t) = \frac{1000}{\mu \tau (t + 1000)}$.}
\label{Fig_Convergence}
\end{figure}

Next we investigate the convergence result, presented in Corollary \ref{CorrConvF_FstarKequalM}, for various $K$ values. We consider two data distribution scenarios, where $(\Gamma, G) \in \{(1, 1), (10, 10)\}$. The case with larger $\Gamma$ and $G$ values models a relatively more biased or less symmetric data distribution, in which case both $\Gamma$ and the variance of the gradient at each device are expected to be larger.

The convergence results for different $K$ values, $K \in \{ 1, 3, 5, 10, 15, 20 \}$, when $M=100$ are simulated for $(\Gamma, G)=(1,1)$ and $(\Gamma, G)=(10,10)$ in Fig. \ref{Fig_Convergence_G1} and Fig. \ref{Fig_Convergence_G100}, respectively. We assume $\sigma^2 = \bar{P} = 1$, $n = 10^5$, $\mu = 1$, $L=5$ and $\left\| \boldsymbol{\theta}(0) - \boldsymbol{\theta}^* \right\|_2^2 = 500$. We set $\tau = 3$ and $\eta (t) = \frac{1000}{\mu \tau (t + 1000)}$, $\forall t$. We observe in Fig. \ref{Fig_Convergence_G1} that, for smaller $\Gamma$ and $G$ values, in which the data distribution across the devices is expected to be more symmetric and less biased, scheduling a single device, i.e., $K=1$, provides the best performance. On the other hand, the results in Fig. \ref{Fig_Convergence_G100} illustrate that, for larger $\Gamma$ and $G$ values, i.e., more biased data distribution, more devices should be scheduled to achieve the best average loss. As can be seen in this figure, $K=1$ provides the best convergence speed, slightly faster than $K=3$. However, $K=1$ has a much higher average loss than $K=3$, $K=5$ and $K=10$. Also, although $K=3$ provides a slightly faster convergence speed, $K=5$ achieves a smaller average loss. 
The convergence results demonstrated in Fig. \ref{Fig_Convergence} corroborate the experimental results shown in Figures \ref{Fig_IID_K_1_5} and \ref{Fig_nonIID2}; that is, when the data is distributed in a more symmetric fashion, such as in the IID scenario, all the resources should be given to a single device, whereas for a more biased data distribution, such as the non-IID scenario, a fraction of devices should share the resources in order to achieve the best performance.              

\begin{figure}[t!]
\centering
\begin{subfigure}{.5\textwidth}
  \centering
  \includegraphics[scale=0.55,trim={8pt 7pt 36pt 40pt},clip]{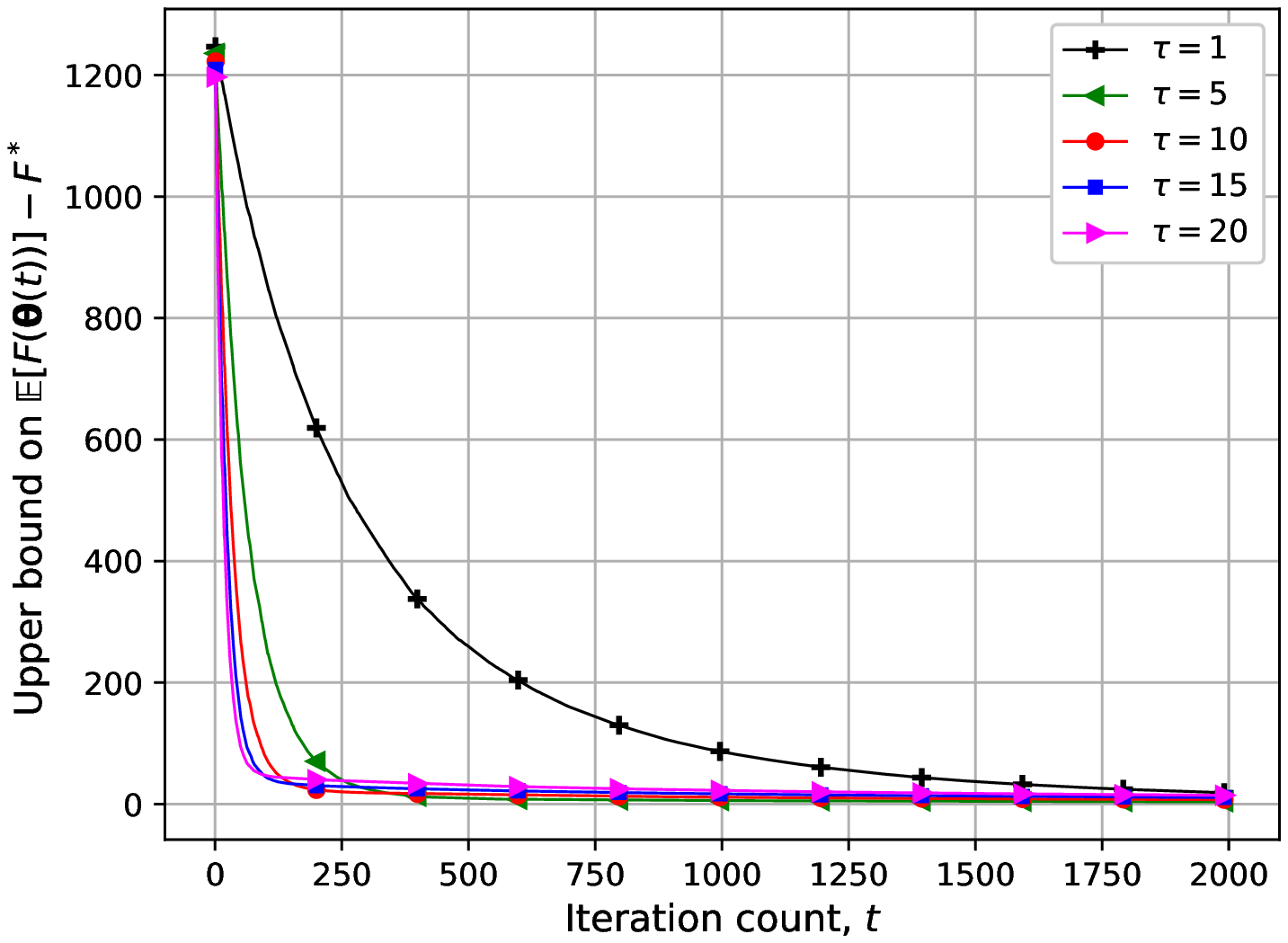}\vspace{0cm}
  \caption{$\Gamma = 1$, $G=1$ and $K=1$}
  \label{Fig_Convergence_Diff_tau_G1}
\end{subfigure}%
\begin{subfigure}{.5\textwidth}
  \centering
  \includegraphics[scale=0.55,trim={8pt 7pt 36pt 40pt},clip]{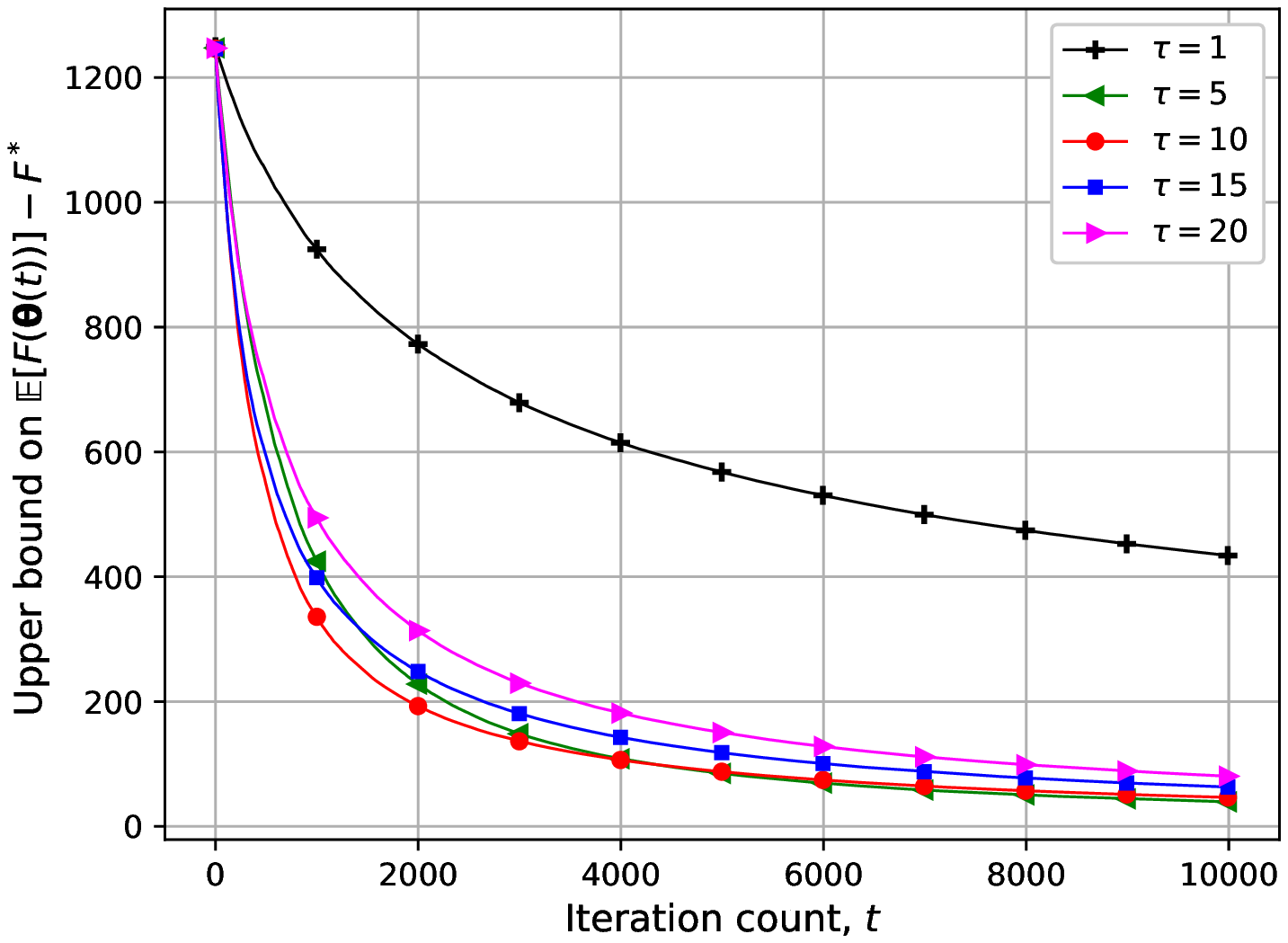}\vspace{0cm}
  \caption{$\Gamma = 10$, $G=10$ and $K=5$}
  \label{Fig_Convergence_Diff_tau_G100}
\end{subfigure}
\caption{Upper bound on $\mathbb{E} \left[ F( \boldsymbol{\theta} (T)) \right] - F^*$ for different $\tau$ values, $\tau \in \{ 1, 5, 10, 15, 20\}$, considering $M=100$, $n = 10^5$, $\mu = 0.25$, $L=5$, $\left\| \boldsymbol{\theta}(0) - \boldsymbol{\theta}^* \right\|_2^2 = 500$ and $\eta (t) = \frac{1000}{20 \mu (t + 1000)}$.}
\label{Fig_Convergence_Diff_tau}
\end{figure}

In Fig. \ref{Fig_Convergence_Diff_tau} we investigate the impact of the number of local iterations $\tau$ on the convergence rate. We again consider two scenarios $(\Gamma, G) = (1, 1)$ and $(\Gamma, G) = (10, 10)$, where the former models a more symmetric data distribution while the latter corresponds to a more biased data distribution, and, based on the observations made in Fig. \ref{Fig_Convergence}, we set $K=1$ and $K=5$, respectively. We consider $\tau \in \{ 2, 3, 5, 10, 15, 20\}$, and for fairness we set $\eta(t) = \frac{1000}{20 \mu (t + 1000)}$ in all the simulations, which satisfies $\eta(t) \le \min \{1, \frac{1}{\mu \tau}\}$. We assume $\sigma^2 = \bar{P} = 1$, $n = 10^5$, $\mu = 0.25$, $L=5$ and $\left\| \boldsymbol{\theta}(0) - \boldsymbol{\theta}^* \right\|_2^2 = 500$. For $(\Gamma, G) = (1, 1)$, we observe that increasing the number of local iterations $\tau = 1$ to $\tau = 5$ significantly improves the convergence performance in terms of both the convergence speed and the final average loss. However, by further increasing $\tau$ the improvement in the convergence speed reduces, and the final average loss increases slightly, which is due to the relatively small asymmetry in the data distribution. This shows that, for a fixed local dataset with limited size, performing an excessive number of local iterations might not help to converge to the optimal solution faster even for a fairly symmetric data distribution. 
Similarly, for $(\Gamma, G) = (10, 10)$, the convergence speed, as well as the average loss improve significantly increasing for $\tau =5$ compared to $\tau =1$. 
However, it is evident that increasing $\tau$ to $\tau = 15$ and $\tau = 20$ deteriorates the performance in terms of both the convergence speed and the final average loss value. 
As expected, designing an efficient $\tau$ value is more crucial for a more biased data distribution, and increasing $\tau$ excessively might cause divergence, since the local datasets do not provide a reliable representation of the whole dataset.

\section{Conclusions}\label{SecConc}
We have studied FL under limited power and bandwidth resources considering block fading channels from the devices to the PS. We have considered orthogonal digital transmissions from the devices to the PS, and studied various scheduling algorithms to decide which devices participate in the learning process at each round.
There is a natural tradeoff between the number of devices participating and the fraction of resources allocated to each device. With more devices scheduled for transmission, the global model parameters are updated at the PS by utilizing a larger fraction of the training data samples; while, each device provides a less accurate estimate of its local model update due to the limited resources available per device.
We have proposed novel device scheduling algorithms that consider not only the channel conditions of the devices, but also the significance of their local model updates. 
We have also established convergence result for FL over power-and bandwidth-limited wireless networks with device scheduling, which, to the best of our knowledge, provides the first convergence result in the literature for the setting under consideration. 
Experiments on the MNIST dataset have shown that it is beneficial to schedule devices based on both their channel conditions and the significance of their model updates rather than considering only one of the two metrics. Also, the best number of participating devices for each considered policy depends on the type of data distribution across devices; for an IID scenario, it is better to schedule a single device, whereas for a non-IID scenario, scheduling a moderate number of devices provides the best performance. The observation on the best number of scheduled devices for IID and non-IID data distribution scenarios is corroborated by the convergence result, where the asymmetry in the data distribution, which corresponds to non-IID scenario, is introduced by increasing the variance of the local gradients.      

\appendices

\section{Proof of Theorem \ref{Theoremtheta_thetastarKequalM}}\label{AppProofTheoremTheta_ThetaStarKequalM}

In addition to the global model parameters updated as
\begin{align}\label{AppFDPModUpdateTheta}
\boldsymbol{\theta} (t+1) = \boldsymbol{\theta} (t) + \frac{1}{K} \sum\nolimits_{m \in \mathcal{M} (t)} \widehat{\boldsymbol{g}}_m (t),
\end{align}
we define the following auxiliary variable:
\begin{align}\label{AppFDPThetaTilde}
\boldsymbol{\upsilon} (t+1) = \boldsymbol{\theta} (t) + \frac{1}{M} \sum\nolimits_{m =1}^{M} \widehat{\boldsymbol{g}}_m (t).
\end{align}
We have 
\begin{align}\label{AppFDP_1}
& \left\| \boldsymbol{\theta} (t+1) - {\boldsymbol{\theta}}^* \right\|_2^2 =  \left\| \boldsymbol{\theta} (t+1) - {\boldsymbol{\upsilon}} (t+1) + {\boldsymbol{\upsilon}} (t+1)  - {\boldsymbol{\theta}}^* \right\|_2^2 \nonumber \\
& \; = \left\| \boldsymbol{\theta} (t+1) - {\boldsymbol{\upsilon}} (t+1) \right\|_2^2 + \left\| {\boldsymbol{\upsilon}} (t+1)  - {\boldsymbol{\theta}}^* \right\|_2^2  + 2 \langle \boldsymbol{\theta} (t+1) - {\boldsymbol{\upsilon}} (t+1) , {\boldsymbol{\upsilon}} (t+1)  - {\boldsymbol{\theta}}^* \rangle.  
\end{align}
In the following, we bound the average of the terms on the right hand side (RHS) of \eqref{AppFDP_1}.

\begin{lemma}\label{AppLemmaTerm_1}
We have
\begin{align}\label{AppLemmaTerm_1_Eq_1}
\mathbb{E} \left[ \left\| \boldsymbol{\theta} (t+1) - {\boldsymbol{\upsilon}} (t+1) \right\|_2^2 \right] \le \frac{(M-K)\rho(t) \eta^2(t) \tau^2 G^2}{K(M-1)}.     
\end{align}
\end{lemma}

\begin{proof}
See Appendix \ref{AppProofPDPLemmaTerm_1}. 
\end{proof}

\begin{lemma}\label{AppFDPLemmaTerm_2}
We have
\begin{align}\label{AppFDPBoundTerm_2}
& \mathbb{E} \left[ \left\| \boldsymbol{\upsilon} (t+1) - {\boldsymbol{\theta}}^* \right\|_2^2 \right] \le \left( 1 - \mu \rho(t) \eta (t) \left( \tau - \eta(t) (\tau - 1) \right) \right) \mathbb{E} \left[ \left\| \boldsymbol{\theta} (t) - {\boldsymbol{\theta}}^* \right\|_2^2 \right] \nonumber\\
& + \rho(t) \left( 1+ \mu (1- \eta(t)) \right) \eta^2(t) G^2 \frac{\tau (\tau-1)(2\tau-1)}{6} + \rho(t) \eta^2(t) (\tau^2 + \tau-1) G^2  + 2 \rho(t) \eta(t) (\tau - 1) \Gamma \nonumber\\
& + 2 \rho(t) \eta (t) \frac{1}{M} \sum\nolimits_{m=1}^{M} \sum\nolimits_{i=2}^{\tau} \left( F_m^* - \mathbb{E} \left[ F_m({\boldsymbol{\theta}}_m^i (t)) \right] \right) + 2 \rho(t) \eta (t) \left(F^* - \mathbb{E} \left[ F({\boldsymbol{\theta}} (t)) \right] \right)  .      
\end{align}
\end{lemma}
\begin{proof}
See Appendix \ref{AppProofFDPLemmaTerm_2}. 
\end{proof}

\begin{lemma}\label{AppPDPLemmaTerm_3}
Let $\mathbb{E}_{\mathcal{M}(t)}$ denote expectation over the device scheduling randomness at the global iteration $t$. We have
\begin{align}\label{AppPDPBoundTerm_3_1}
\mathbb{E}_{\mathcal{M}(t)} \left[ \boldsymbol{\theta} (t+1) \right] = {\boldsymbol{\upsilon}} (t+1),   
\end{align}
from which it follows that
\begin{align}\label{AppPDP_2}
\mathbb{E}_{\mathcal{M}(t)} \left[ \langle \boldsymbol{\theta} (t+1) - {\boldsymbol{\upsilon}} (t+1) , {\boldsymbol{\upsilon}} (t+1)  - {\boldsymbol{\theta}}^* \rangle \right] = 0.   
\end{align}
\end{lemma}
\begin{proof}
Due to the randomness of the device scheduling policy, it follows that
\newcommand\thirdequal{\mathrel{\overset{\makebox[0pt]{\mbox{\normalfont\tiny\sffamily (a)}}}{=}}}
\begin{align}\label{AppPDPBoundTerm_3_2}
\mathbb{E}_{\mathcal{M}(t)} \left[ \frac{1}{K} \sum\nolimits_{m \in \mathcal{M} (t)} \widehat{\boldsymbol{g}}_m (t) \right] \thirdequal \frac{\binom{M-1}{K-1}}{K\binom{M}{K}} \sum\nolimits_{m=1}^{M} \widehat{\boldsymbol{g}}_m (t) = \frac{1}{M} \sum\nolimits_{m=1}^{M} \widehat{\boldsymbol{g}}_m (t),       
\end{align}
where (a) follows from device scheduling randomness and the fact that the quantized model update of each device appears $\binom{M-1}{K-1}$ times. The proof of Lemma \ref{AppPDPLemmaTerm_3} is concluded from \eqref{AppPDPBoundTerm_3_2}. 
\end{proof}

According to the results in Lemmas \ref{AppLemmaTerm_1}, \ref{AppFDPLemmaTerm_2} and \ref{AppPDPLemmaTerm_3}, it follows that 
\newcommand\firstinequal{\mathrel{\overset{\makebox[0pt]{\mbox{\normalfont\tiny\sffamily (a)}}}{\le}}}
\newcommand\secondinequal{\mathrel{\overset{\makebox[0pt]{\mbox{\normalfont\tiny\sffamily (b)}}}{\le}}}
\begin{align}\label{AppFDP_2}
&\mathbb{E} \left[ \left\| \boldsymbol{\theta} (t+1) - {\boldsymbol{\theta}}^* \right\|_2^2 \right] \le \left( 1 - \mu \rho(t) \eta (t) \left( \tau - \eta(t) (\tau - 1) \right) \right) \mathbb{E} \left[ \left\| \boldsymbol{\theta} (t) - {\boldsymbol{\theta}}^* \right\|_2^2 \right] \nonumber\\
& \;\; \qquad \quad + \frac{(M-K)\rho(t) \eta^2(t) \tau^2 G^2}{K(M-1)} + \rho(t) \eta^2(t) \left( \tau^2 + \tau - 1 \right) G^2  \nonumber\\
& \;\; \qquad \quad + \rho(t) \left( 1+ \mu (1- \eta(t)) \right) \eta^2(t) G^2 \frac{\tau (\tau-1)(2\tau-1)}{6}  + 2 \rho(t) \eta(t) (\tau - 1) \Gamma \nonumber\\
& \;\; \qquad \quad + 2 \rho(t) \eta (t) \frac{1}{M} \sum\nolimits_{m=1}^{M} \sum\nolimits_{i=2}^{\tau} \left( F_m^* - \mathbb{E} \left[ F_m({\boldsymbol{\theta}}_m^i (t)) \right] \right) + 2 \rho(t) \eta (t) \left(F^* - \mathbb{E} \left[ F({\boldsymbol{\theta}} (t)) \right] \right)  \nonumber\\
& \qquad \; \firstinequal \left( 1 - \mu \rho(t) \eta (t) \left( \tau - \eta(t) (\tau - 1) \right) \right) \mathbb{E} \left[ \left\| \boldsymbol{\theta} (t) - {\boldsymbol{\theta}}^* \right\|_2^2 \right] \nonumber\\
& \;\;  \qquad \quad + \frac{(M-K)\rho(t) \eta^2(t) \tau^2 G^2}{K(M-1)} + \rho(t) \eta^2(t) \left( \tau^2 + \tau -1 \right) G^2  \nonumber\\
& \;\;  \qquad \quad + \rho(t) \left( 1+ \mu (1- \eta(t)) \right) \eta^2(t) G^2 \frac{\tau (\tau-1)(2\tau-1)}{6} + 2 \rho(t) \eta(t) (\tau - 1) \Gamma,
\end{align}
where (a) follows since $F^* - F(\boldsymbol{\theta} (t)) \le 0$, $\forall t$, and $F_m^* - F_m({\boldsymbol{\theta}}_m^i (t)) \le 0$, $\forall m, t$. Theorem \ref{Theoremtheta_thetastarKequalM} is concluded from the inequality in \eqref{AppFDP_2}. 


\section{Proof of Lemma \ref{AppLemmaTerm_1}}\label{AppProofPDPLemmaTerm_1}
To prove Lemma \ref{AppLemmaTerm_1}, we take similar steps as \cite[Appendix B.4]{XLiFedAveFLnonIID}. We have 
\begin{align}\label{AppPDPLemmaTemr_1_Eq_1}
\mathbb{E} \left[ \left\| \boldsymbol{\theta} (t+1) - {\boldsymbol{\upsilon}} (t+1) \right\|_2^2 \right] =  \mathbb{E} \bigg[ \left\| \frac{1}{K} \sum\nolimits_{m \in \mathcal{M} (t)} \widehat{\boldsymbol{g}}_m (t) - \widehat{\boldsymbol{g}} (t) \right\|_2^2 \bigg],   
\end{align}
where we have defined 
\begin{align}
\widehat{\boldsymbol{g}} (t) \triangleq \frac{1}{M} \sum\nolimits_{m =1}^{M} \widehat{\boldsymbol{g}}_m (t).    
\end{align}
We also denote the indicator function by $\mathds{1} (\cdot)$. We have
\newcommand\thirdequal{\mathrel{\overset{\makebox[0pt]{\mbox{\normalfont\tiny\sffamily (a)}}}{=}}}
\begin{align}\label{AppPDPLemmaTemr_1_Eq_2}
&\mathbb{E} \left[ \left\| \boldsymbol{\theta} (t+1) - {\boldsymbol{\upsilon}} (t+1) \right\|_2^2 \right]  = \mathbb{E} \bigg[ \left\| \frac{1}{K} \sum\nolimits_{m =1}^{M} \mathds{1} (m \in \mathcal{M} (t)) \left( \widehat{\boldsymbol{g}}_m (t) - \widehat{\boldsymbol{g}} (t) \right) \right\|_2^2 \bigg]\nonumber\\
& \; = \frac{1}{K^2} \mathbb{E} \left[ \sum\nolimits_{m =1}^{M} \mathds{1} (m \in \mathcal{M} (t)) \left\| \widehat{\boldsymbol{g}}_m (t) - \widehat{\boldsymbol{g}} (t) \right\|_2^2 \right.\nonumber\\
&\; \quad \; \left. + \sum\nolimits_{m =1}^{M} \sum\nolimits_{m' =1, m' \ne m}^{M} \mathds{1} (m \in \mathcal{M} (t)) \mathds{1} (m' \in \mathcal{M} (t)) \langle \widehat{\boldsymbol{g}}_m (t) - \widehat{\boldsymbol{g}} (t), \widehat{\boldsymbol{g}}_{m'} (t) - \widehat{\boldsymbol{g}} (t) \rangle \right].   
\end{align}
Due to the symmetry, it follows that 
\begin{align}\label{AppPDPLemmaTemr_1_Eq_3}
\mathbb{E}_{\mathcal{M}(t)} \bigg[ \sum\limits_{m =1}^{M} \mathds{1} (m \in \mathcal{M} (t)) \left\| \widehat{\boldsymbol{g}}_m (t) - \widehat{\boldsymbol{g}} (t) \right\|_2^2 \bigg] &\thirdequal \frac{\binom{M-1}{K-1}}{\binom{M}{K}} \sum\limits_{m =1}^{M} \left\| \widehat{\boldsymbol{g}}_m (t) - \widehat{\boldsymbol{g}} (t) \right\|_2^2 \nonumber\\
& = \frac{K}{M} \sum\limits_{m =1}^{M} \left\| \widehat{\boldsymbol{g}}_m (t) - \widehat{\boldsymbol{g}} (t) \right\|_2^2,   
\end{align}
where (a) is due to the fact that each index $m$, for $m \in [M]$, appears $\binom{M-1}{K-1}$ times with the term on the left hand side (LHS) of the equality, and 
\newcommand\fifthequal{\mathrel{\overset{\makebox[0pt]{\mbox{\normalfont\tiny\sffamily (b)}}}{=}}}
\begin{align}\label{AppPDPLemmaTemr_1_Eq_4}
& \mathbb{E}_{\mathcal{M}(t)} \bigg[ \sum\limits_{m =1}^{M} \sum\limits_{m' =1, m' \ne m}^{M} \mathds{1} (m \in \mathcal{M} (t)) \mathds{1} (m' \in \mathcal{M} (t)) \langle \widehat{\boldsymbol{g}}_m (t) - \widehat{\boldsymbol{g}} (t), \widehat{\boldsymbol{g}}_{m'} (t) - \widehat{\boldsymbol{g}} (t) \rangle \bigg] \nonumber\\
& \qquad \qquad \qquad \qquad \qquad \fifthequal \frac{\binom{M-2}{K-2}}{\binom{M}{K}} \sum\limits_{m =1}^{M} \sum\limits_{m' =1, m' \ne m}^{M} \langle \widehat{\boldsymbol{g}}_m (t) - \widehat{\boldsymbol{g}} (t), \widehat{\boldsymbol{g}}_{m'} (t) - \widehat{\boldsymbol{g}} (t) \rangle \nonumber\\
& \qquad \qquad \qquad \qquad \qquad = \frac{K(K-1)}{M(M-1)} \sum\limits_{m =1}^{M} \sum\limits_{m' =1, m' \ne m}^{M} \langle \widehat{\boldsymbol{g}}_m (t) - \widehat{\boldsymbol{g}} (t), \widehat{\boldsymbol{g}}_{m'} (t) - \widehat{\boldsymbol{g}} (t) \rangle,   
\end{align}
where (b) follows since each specific index pair $(m, m')$, for $m, m' \in [M]$, $m \ne m'$, appears $\binom{M-2}{K-2}$ times on the LHS of the equality. Substituting \eqref{AppPDPLemmaTemr_1_Eq_3} and \eqref{AppPDPLemmaTemr_1_Eq_4} into \eqref{AppPDPLemmaTemr_1_Eq_2} yields
\newcommand\sixthequal{\mathrel{\overset{\makebox[0pt]{\mbox{\normalfont\tiny\sffamily (c)}}}{=}}}
\newcommand\seventhequal{\mathrel{\overset{\makebox[0pt]{\mbox{\normalfont\tiny\sffamily (d)}}}{=}}}
\newcommand\sixthinequal{\mathrel{\overset{\makebox[0pt]{\mbox{\normalfont\tiny\sffamily (e)}}}{\le}}}
\newcommand\seventhinequal{\mathrel{\overset{\makebox[0pt]{\mbox{\normalfont\tiny\sffamily (f)}}}{\le}}}
\begin{align}\label{AppPDPLemmaTemr_1_Eq_5}
& \mathbb{E} \left[ \left\| \boldsymbol{\theta} (t+1) - {\boldsymbol{\upsilon}} (t+1) \right\|_2^2 \right] = \frac{1}{KM}  \sum\nolimits_{m =1}^{M} \mathbb{E} \left[ \left\| \widehat{\boldsymbol{g}}_m (t) - \widehat{\boldsymbol{g}} (t) \right\|_2^2 \right]  \nonumber\\
& \qquad \;\;\;\; + \frac{K-1}{KM(M-1)} \sum\nolimits_{m =1}^{M} \sum\nolimits_{m' =1, m' \ne m}^{M} \mathbb{E} \left[ \langle \widehat{\boldsymbol{g}}_m (t) - \widehat{\boldsymbol{g}} (t), \widehat{\boldsymbol{g}}_{m'} (t) - \widehat{\boldsymbol{g}} (t) \rangle \right] \nonumber\\
& \qquad \sixthequal \frac{M-K}{KM(M-1)} \sum\nolimits_{m =1}^{M} \mathbb{E} \left[ \left\| \widehat{\boldsymbol{g}}_m (t) - \widehat{\boldsymbol{g}} (t) \right\|_2^2 \right]\nonumber\\
& \qquad = \frac{M-K}{MK(M-1)} \left( \sum\nolimits_{m =1}^{M} \mathbb{E} \left[ \left\| \widehat{\boldsymbol{g}}_m (t) \right\|_2^2 \right] - \mathbb{E} \left[ \left\| \widehat{\boldsymbol{g}} (t) \right\|_2^2 \right] \right) \nonumber\\
& \qquad \le \frac{M-K}{MK(M-1)}  \sum\nolimits_{m =1}^{M} \mathbb{E} \left[ \left\| \widehat{\boldsymbol{g}}_m (t) \right\|_2^2 \right] \nonumber\\
& \qquad \seventhequal \frac{(M-K) \rho(t)}{MK(M-1)} \sum\nolimits_{m =1}^{M} \mathbb{E} \left[ \left\| {\boldsymbol{g}}_m (t) \right\|_2^2 \right] \nonumber\\
& \qquad = \frac{(M-K) \rho(t)\eta^2(t)}{MK(M-1)} \sum\nolimits_{m =1}^{M} \mathbb{E} \left[ \left\| \sum\nolimits_{i=1}^{\tau} \nabla F_m \left( \boldsymbol{\theta}_m^i (t), \xi_m^i (t) \right) \right\|_2^2 \right] \nonumber \\
& \qquad \sixthinequal \frac{(M-K) \rho(t)\eta^2(t) \tau}{MK(M-1)} \sum\nolimits_{m =1}^{M} \sum\nolimits_{i=1}^{\tau} \mathbb{E} \left[ \left\| \nabla F_m \left( \boldsymbol{\theta}_m^i (t), \xi_m^i (t) \right) \right\|_2^2 \right] \nonumber\\
& \qquad \seventhinequal \frac{(M-K) \rho(t) \eta^2(t) \tau^2 G^2}{K(M-1)},
\end{align}
where (c) follows since 
\begin{align}
\left\| \sum\nolimits_{m =1}^{M} \left( \widehat{\boldsymbol{g}}_m (t) - \widehat{\boldsymbol{g}} (t) \right) \right\|_2^2 = 0,   
\end{align}
(d) follows from Lemma \ref{LemmaMeanVarhatg}, (e) follows from the convexity of $\| \cdot \|_2^2$, and (f) follows from Assumption \ref{AssumpBoundedVarGradient}. This completes the proof of Lemma \ref{AppLemmaTerm_1}.

\section{Proof of Lemma \ref{AppFDPLemmaTerm_2}}\label{AppProofFDPLemmaTerm_2}

We have
\begin{align}\label{AppLemmaTemr_2_Eq_1}
&\mathbb{E} \left[ \left\| \boldsymbol{\upsilon} (t+1) - {\boldsymbol{\theta}}^* \right\|_2^2 \right] = \mathbb{E} \bigg[ \left\| \boldsymbol{\theta} (t) + \frac{1}{M} \sum\nolimits_{m =1}^{M} \widehat{\boldsymbol{g}}_m (t) - {\boldsymbol{\theta}}^* \right\|_2^2 \bigg] 
\end{align}

\begin{align}
= \mathbb{E} \left[ \left\| \boldsymbol{\theta} (t) - {\boldsymbol{\theta}}^* \right\|_2^2 \right] + \mathbb{E} \bigg[ \left\| \frac{1}{M} \sum\nolimits_{m =1}^{M} \widehat{\boldsymbol{g}}_m (t) \right\|_2^2 \bigg] + 2 \mathbb{E} \left[ \langle \boldsymbol{\theta} (t) - {\boldsymbol{\theta}}^* , \frac{1}{M} \sum\nolimits_{m =1}^{M} \widehat{\boldsymbol{g}}_m (t) \rangle \right].
\end{align}
The convexity of $\left\| \cdot \right\|_2^2$ results in
\newcommand\onestequal{\mathrel{\overset{\makebox[0pt]{\mbox{\normalfont\tiny\sffamily (a)}}}{=}}}
\begin{align}\label{AppLemmaTemr_2_Eq_1_2}
&\mathbb{E} \bigg[ \left\| \frac{1}{M} \sum\nolimits_{m =1}^{M} \widehat{\boldsymbol{g}}_m (t) \right\|_2^2 \bigg] \le  \frac{1}{M} \sum\nolimits_{m =1}^{M} \mathbb{E} \left[ \left\| \widehat{\boldsymbol{g}}_m (t) \right\|_2^2 \right] \onestequal \frac{\rho(t)}{M} \sum\nolimits_{m =1}^{M} \mathbb{E} \left[ \left\| {\boldsymbol{g}}_m (t) \right\|_2^2 \right] \nonumber\\
& \qquad \qquad \quad = \frac{\rho(t)\eta^2(t)}{M} \sum\nolimits_{m =1}^{M} \mathbb{E} \left[ \left\| \sum\nolimits_{i=1}^{\tau} \nabla F_m \left( \boldsymbol{\theta}_m^i (t), \xi_m^i (t) \right) \right\|_2^2 \right] \nonumber\\
& \qquad \qquad \quad  \le \frac{\rho(t) \eta^2(t) \tau}{M} \sum\nolimits_{m =1}^{M} \sum\nolimits_{i=1}^{\tau} \mathbb{E} \left[ \left\| \nabla F_m \left( \boldsymbol{\theta}_m^i (t), \xi_m^i (t) \right) \right\|_2^2 \right] \le \rho(t) \eta^2(t) \tau^2 G^2,
\end{align}
where (a) follows from \eqref{Meanhatg}. Plugging the above result into \eqref{AppLemmaTemr_2_Eq_1} yields
\begin{align}\label{AppLemmaTemr_2_Eq_1_3}
&\mathbb{E} \left[ \left\| \boldsymbol{\upsilon} (t+1) - {\boldsymbol{\theta}}^* \right\|_2^2 \right] \le \mathbb{E} \left[ \left\| \boldsymbol{\theta} (t) - {\boldsymbol{\theta}}^* \right\|_2^2 \right] + \rho(t) \eta^2(t) \tau^2 G^2 + 2 \mathbb{E} \bigg[ \langle \boldsymbol{\theta} (t) - {\boldsymbol{\theta}}^* , \frac{1}{M} \sum\limits_{m =1}^{M} \widehat{\boldsymbol{g}}_m (t) \rangle \bigg].
\end{align}
In the following, we bound the last term on the RHS of the above inequality. We have
\begin{align}\label{AppLemmaTemr_2_Eq_2}
&2 \mathbb{E} \left[ \langle \boldsymbol{\theta} (t) - {\boldsymbol{\theta}}^* , \frac{1}{M} \sum\nolimits_{m =1}^{M} \widehat{\boldsymbol{g}}_m (t) \rangle \right] \onestequal \frac{2 \rho (t)}{M} \sum\nolimits_{m=1}^{M} \mathbb{E} \left[ \langle \boldsymbol{\theta} (t) - {\boldsymbol{\theta}}^* , {\boldsymbol{g}}_m (t) \rangle \right] \nonumber\\
& \qquad = \frac{2 \rho (t) \eta(t)}{M} \sum\nolimits_{m=1}^{M} \mathbb{E} \left[ \langle {\boldsymbol{\theta}}^* - \boldsymbol{\theta} (t) , \sum\nolimits_{i=1}^{\tau} \nabla F_m \left( \boldsymbol{\theta}_m^i (t), \xi_m^i (t) \right) \rangle \right] \nonumber\\
& \qquad = \frac{2 \rho (t) \eta(t)}{M} \sum\nolimits_{m=1}^{M} \mathbb{E} \left[ \langle {\boldsymbol{\theta}}^* - \boldsymbol{\theta} (t) , \nabla F_m \left( \boldsymbol{\theta} (t), \xi_m^1 (t) \right) \rangle \right] \nonumber
\end{align}

\begin{align}
& \qquad \quad + \frac{2 \rho (t) \eta(t)}{M} \sum\nolimits_{m=1}^{M} \mathbb{E} \left[ \langle {\boldsymbol{\theta}}^* - \boldsymbol{\theta} (t) , \sum\nolimits_{i=2}^{\tau} \nabla F_m \left( \boldsymbol{\theta}_m^i (t), \xi_m^i (t) \right) \rangle \right], 
\end{align}
where (a) follows from Lemma \ref{LemmaMeanVarhatg}. 
In the following we bound the two terms on the RHS of the equality in \eqref{AppLemmaTemr_2_Eq_2}. 
We have
\newcommand\thirdinequal{\mathrel{\overset{\makebox[0pt]{\mbox{\normalfont\tiny\sffamily (b)}}}{\le}}}
\begin{align}\label{AppLemmaTemr_2_Eq_3}
& \frac{2 \rho (t) \eta(t)}{M} \sum\nolimits_{m=1}^{M} \mathbb{E} \left[ \langle {\boldsymbol{\theta}}^* - \boldsymbol{\theta} (t) , \nabla F_m \left( \boldsymbol{\theta} (t), \xi_m^1 (t) \right) \rangle \right] \nonumber\\
& \qquad \qquad \qquad \qquad \thirdequal \frac{2 \rho (t) \eta(t)}{M} \sum\nolimits_{m=1}^{M} \mathbb{E} \left[ \langle {\boldsymbol{\theta}}^* - \boldsymbol{\theta} (t) , \nabla F_m \left( \boldsymbol{\theta} (t) \right) \rangle \right] \nonumber\\
& \qquad \qquad \qquad \qquad \thirdinequal \frac{2 \rho (t) \eta(t)}{M} \sum\nolimits_{m=1}^{M} \mathbb{E} \left[ F_m (\boldsymbol{\theta}^*) - F_m(\boldsymbol{\theta} (t)) - \frac{\mu}{2} \left\| \boldsymbol{\theta} (t) - {\boldsymbol{\theta}}^* \right\|_2^2 \right] \nonumber\\
& \qquad \qquad \qquad \qquad = 2 \rho (t) \eta(t) \left( F^* - \mathbb{E} \left[ F(\boldsymbol{\theta} (t)) \right] - \frac{\mu}{2} \mathbb{E} \left[ \left\| \boldsymbol{\theta} (t) - {\boldsymbol{\theta}}^* \right\|_2^2 \right] \right),
\end{align}
where (a) follows since $\mathbb{E}_{\xi} \left[ \nabla F_m \left( \boldsymbol{\theta} (t), \xi^1_m (t) \right) \right] = \nabla F_m \left( \boldsymbol{\theta} (t)  \right)$, $\forall m, t$, and (b) follows since $F_m$ is $\mu$-strongly convex. For the second term on the RHS of the equality in \eqref{AppLemmaTemr_2_Eq_2}, we have  
\begin{align}\label{AppLemmaTemr_2_Eq_4}
& \frac{2 \rho (t) \eta(t)}{M} \sum\nolimits_{m=1}^{M} \mathbb{E} \left[ \langle {\boldsymbol{\theta}}^* - \boldsymbol{\theta} (t) , \sum\nolimits_{i=2}^{\tau} \nabla F_m \left( \boldsymbol{\theta}_m^i (t), \xi_m^i (t) \right) \rangle \right] \nonumber\\
& \qquad \qquad \qquad \qquad = \frac{2 \rho (t) \eta(t)}{M} \sum\nolimits_{m=1}^{M} \sum\nolimits_{i=2}^{\tau} \mathbb{E} \left[ \langle {\boldsymbol{\theta}}^* - \boldsymbol{\theta} (t) , \nabla F_m \left( \boldsymbol{\theta}_m^i (t), \xi_m^i (t) \right) \rangle \right] \nonumber\\
& \qquad \qquad \qquad \qquad = \frac{2 \rho (t) \eta(t)}{M} \sum\nolimits_{m=1}^{M} \sum\nolimits_{i=2}^{\tau} \mathbb{E} \left[ \langle \boldsymbol{\theta}_m^i (t) - \boldsymbol{\theta} (t) , \nabla F_m \left( \boldsymbol{\theta}_m^i (t), \xi_m^i (t) \right) \rangle \right] \nonumber\\
& \qquad \qquad \qquad \qquad \quad + \frac{2 \rho (t) \eta(t)}{M} \sum\nolimits_{m=1}^{M} \sum\nolimits_{i=2}^{\tau} \mathbb{E} \left[ \langle \boldsymbol{\theta}^* - \boldsymbol{\theta}_m^i (t) , \nabla F_m \left( \boldsymbol{\theta}_m^i (t), \xi_m^i (t) \right) \rangle \right].
\end{align}
From Cauchy-Schwarz inequality, it follows that
\begin{align}\label{AppLemmaTemr_2_Eq_5}
& \frac{2 \rho (t) \eta(t)}{M} \sum\nolimits_{m=1}^{M} \sum\nolimits_{i=2}^{\tau} \mathbb{E} \left[ \langle \boldsymbol{\theta}_m^i (t) - \boldsymbol{\theta} (t) , \nabla F_m \left( \boldsymbol{\theta}_m^i (t), \xi_m^i (t) \right) \rangle \right] \nonumber\\
& \; \; \; \quad \le \frac{\rho (t) \eta(t)}{M} \sum\nolimits_{m=1}^{M} \sum\nolimits_{i=2}^{\tau} \mathbb{E} \bigg[ \frac{1}{\eta(t)} \left\| \boldsymbol{\theta}_m^i (t) - \boldsymbol{\theta} (t) \right\|_2^2 + \eta(t) \left\| \nabla F_m \left( \boldsymbol{\theta}_m^i (t), \xi_m^i (t) \right) \right\|_2^2 \bigg]\nonumber\\
& \; \; \; \quad \firstinequal \frac{\rho (t)}{M} \sum\nolimits_{m=1}^{M} \sum\nolimits_{i=2}^{\tau} \mathbb{E} \left[ \left\| \boldsymbol{\theta}_m^i (t) - \boldsymbol{\theta} (t) \right\|_2^2 \right] + \rho (t) \eta^2 (t) \left( \tau - 1 \right) G^2, 
\end{align}
where (a) follows from Assumption \ref{AssumpBoundedVarGradient}.  
Also, the following lemma presents an upper bound on the second term in the RHS of \eqref{AppLemmaTemr_2_Eq_4}.

\begin{lemma}\label{LemmaTermE}
The second term in the RHS of \eqref{AppLemmaTemr_2_Eq_4} is upper bounded as follows:
\begin{align}\label{EQ_LemmaTermE}
& \frac{2 \rho (t) \eta(t)}{M} \sum\nolimits_{m=1}^{M} \sum\nolimits_{i=2}^{\tau} \mathbb{E} \left[ \langle \boldsymbol{\theta}^* - \boldsymbol{\theta}_m^i (t) , \nabla F_m \left( \boldsymbol{\theta}_m^i (t), \xi_m^i (t) \right) \rangle \right] \nonumber\\
& \qquad \qquad  \le - \mu \rho(t) \eta(t) (1 - \eta(t)) (\tau - 1) \mathbb{E} \left[ \left\| \boldsymbol{\theta} (t) - \boldsymbol{\theta}^* \right\|_2^2 \right] \nonumber\\
& \qquad \qquad \quad + \frac{\mu \rho(t) (1- \eta(t))}{M} \sum\nolimits_{m=1}^{M} \sum\nolimits_{i=2}^{\tau} \mathbb{E} \left[ \left\| \boldsymbol{\theta}_m^i (t) - \boldsymbol{\theta} (t) \right\|_2^2 \right] \nonumber\\
& \qquad \qquad \quad + 2 \rho(t) \eta(t) (\tau - 1) \Gamma + \frac{2 \rho(t) \eta (t)}{M} \sum\nolimits_{m=1}^{M} \sum\nolimits_{i=2}^{\tau} \left( F_m^* - \mathbb{E} \left[ F_m({\boldsymbol{\theta}}_m^i (t)) \right] \right).
\end{align}
\end{lemma}
\begin{proof}
See Appendix \ref{AppProofLemmaTermE}. 
\end{proof}

Substituting the results in \eqref{AppLemmaTemr_2_Eq_5} and \eqref{EQ_LemmaTermE} into \eqref{AppLemmaTemr_2_Eq_4} yields
\begin{align}\label{AppLemmaTemr_2_Eq_6}
& \frac{2 \rho (t) \eta(t)}{M} \sum\nolimits_{m=1}^{M} \mathbb{E} \left[ \langle {\boldsymbol{\theta}}^* - \boldsymbol{\theta} (t) , \sum\nolimits_{i=2}^{\tau} \nabla F_m \left( \boldsymbol{\theta}_m^i (t), \xi_m^i (t) \right) \rangle \right] \nonumber\\
& \qquad \qquad  \le - \mu \rho(t) \eta(t) (1 - \eta(t)) (\tau - 1) \mathbb{E} \left[ \left\| \boldsymbol{\theta} (t) - \boldsymbol{\theta}^* \right\|_2^2 \right] \nonumber\\
& \qquad \qquad  \quad + \frac{\rho(t) \left( 1+ \mu (1- \eta(t)) \right)}{M} \sum\nolimits_{m=1}^{M} \sum\nolimits_{i=2}^{\tau} \mathbb{E} \left[ \left\| \boldsymbol{\theta}_m^i (t) - \boldsymbol{\theta} (t) \right\|_2^2 \right] + \rho (t) \eta^2 (t) \left( \tau - 1 \right) \nonumber\\
& \qquad \qquad  \quad + 2 \rho(t) \eta(t) (\tau - 1) \Gamma +  \frac{2 \rho(t) \eta (t)}{M} \sum\nolimits_{m=1}^{M} \sum\nolimits_{i=2}^{\tau} \left( F_m^* - \mathbb{E} \left[ F_m({\boldsymbol{\theta}}_m^i (t)) \right] \right).
\end{align}
We have 
\begin{align}\label{AppLemmaTemr_2_Eq_7}
&\frac{1}{M} \sum\nolimits_{m=1}^{M} \sum\nolimits_{i=2}^{\tau} \mathbb{E} \left[ \left\| \boldsymbol{\theta}_m^i (t) - \boldsymbol{\theta} (t) \right\|_2^2 \right] \nonumber\\
& = \frac{\eta^2(t)}{M} \sum\nolimits_{m=1}^{M} \sum\nolimits_{i=2}^{\tau} \mathbb{E} \left[ \left\| \sum\nolimits_{j=1}^{i} \nabla F_m \left( \boldsymbol{\theta}_m^j (t), \xi_m^j (t) \right) \right\|_2^2 \right] \firstinequal \eta^2(t) G^2 \frac{\tau (\tau-1)(2\tau-1)}{6}, 
\end{align}
where (a) follows from the convexity of $\left\| \cdot \right\|_2^2$ and Assumption \ref{AssumpBoundedVarGradient}. Having $\eta(t) \le 1$, $\forall t$, from \eqref{AppLemmaTemr_2_Eq_6} and \eqref{AppLemmaTemr_2_Eq_7}, it follows that
\begin{align}\label{AppLemmaTemr_2_Eq_8}
& \frac{2 \rho (t) \eta(t)}{M} \sum\nolimits_{m=1}^{M} \mathbb{E} \left[ \langle {\boldsymbol{\theta}}^* - \boldsymbol{\theta} (t) , \sum\nolimits_{i=2}^{\tau} \nabla F_m \left( \boldsymbol{\theta}_m^i (t), \xi_m^i (t) \right) \rangle \right] \nonumber\\
& \qquad \qquad  \le - \mu \rho(t) \eta(t) (1 - \eta(t)) (\tau - 1) \mathbb{E} \left[ \left\| \boldsymbol{\theta} (t) - \boldsymbol{\theta}^* \right\|_2^2 \right] \nonumber\\
& \qquad \qquad  \quad + \rho(t) \left( 1+ \mu (1- \eta(t)) \right) \eta^2(t) G^2 \frac{\tau (\tau-1)(2\tau-1)}{6} + \rho (t) \eta^2 (t) \left( \tau - 1 \right) G^2 \nonumber\\
& \qquad \qquad  \quad + 2 \rho(t) \eta(t) (\tau - 1) \Gamma + 2 \rho(t) \eta (t) \frac{1}{M} \sum\nolimits_{m=1}^{M} \sum\nolimits_{i=2}^{\tau} \left( F_m^* - \mathbb{E} \left[ F_m({\boldsymbol{\theta}}_m^i (t)) \right] \right).
\end{align}
By substituting the results in \eqref{AppLemmaTemr_2_Eq_3} and \eqref{AppLemmaTemr_2_Eq_8} into \eqref{AppLemmaTemr_2_Eq_2}, we obtain
\begin{align}\label{AppLemmaTemr_2_Eq_9}
& 2 \mathbb{E} \left[ \langle \boldsymbol{\theta} (t) - {\boldsymbol{\theta}}^* , \frac{1}{M} \sum\nolimits_{m =1}^{M} \widehat{\boldsymbol{g}}_m (t) \rangle \right]  \le - \mu \rho(t) \eta(t) \left( \tau - \eta(t) (\tau - 1) \right) \mathbb{E} \left[ \left\| \boldsymbol{\theta} (t) - \boldsymbol{\theta}^* \right\|_2^2 \right] \nonumber\\
& \;  + \rho(t) \left( 1+ \mu (1- \eta(t)) \right) \eta^2(t) G^2 \frac{\tau (\tau-1)(2\tau-1)}{6} + \rho (t) \eta^2 (t) \left( \tau - 1 \right)G^2 + 2 \rho(t) \eta(t) (\tau - 1) \Gamma \nonumber\\
& \;  + 2 \rho(t) \eta (t) \frac{1}{M} \sum\nolimits_{m=1}^{M} \sum\nolimits_{i=2}^{\tau} \left( F_m^* - \mathbb{E} \left[ F_m({\boldsymbol{\theta}}_m^i (t)) \right] \right) + 2 \rho (t) \eta(t) \left( F^* - \mathbb{E} \left[ F(\boldsymbol{\theta} (t)) \right] \right).
\end{align}
Plugging \eqref{AppLemmaTemr_2_Eq_9} into \eqref{AppLemmaTemr_2_Eq_1_3} completes the proof of Lemma \ref{AppFDPLemmaTerm_2}.

\section{Proof of Lemma \ref{LemmaTermE}}\label{AppProofLemmaTermE}
We have 
\begin{align}\label{EQ_LemmaTermE_app_1}
& \frac{2 \rho (t) \eta(t)}{M} \sum\nolimits_{m=1}^{M} \sum\nolimits_{i=2}^{\tau} \mathbb{E} \left[ \langle \boldsymbol{\theta}^* - \boldsymbol{\theta}_m^i (t) , \nabla F_m \left( \boldsymbol{\theta}_m^i (t), \xi_m^i (t) \right) \rangle \right] \nonumber\\
& \thirdequal \frac{2 \rho (t) \eta(t)}{M} \sum\nolimits_{m=1}^{M} \sum\nolimits_{i=2}^{\tau} \mathbb{E} \left[ \langle \boldsymbol{\theta}^* - \boldsymbol{\theta}_m^i (t) , \nabla F_m \left( \boldsymbol{\theta}_m^i (t) \right) \rangle \right]\nonumber\\
& \thirdinequal \frac{2 \rho (t) \eta(t)}{M} \sum\nolimits_{m=1}^{M} \sum\nolimits_{i=2}^{\tau} \mathbb{E} \left[ F_m (\boldsymbol{\theta}^*) - F_m(\boldsymbol{\theta}_m^i (t)) - \frac{\mu}{2} \left\| \boldsymbol{\theta}_m^i (t) - {\boldsymbol{\theta}}^* \right\|_2^2 \right]\nonumber\\
& = \frac{2 \rho (t) \eta(t)}{M} \sum\nolimits_{m=1}^{M} \sum\nolimits_{i=2}^{\tau} \mathbb{E} \left[ F_m (\boldsymbol{\theta}^*) - F_m^* + F_m^* - F_m(\boldsymbol{\theta}_m^i (t)) - \frac{\mu}{2} \left\| \boldsymbol{\theta}_m^i (t) - {\boldsymbol{\theta}}^* \right\|_2^2 \right]\nonumber\\
& = 2 \rho (t) \eta(t) (\tau - 1) \Big( F^* - \frac{1}{M} \sum\nolimits_{m=1}^{M} F_m^* \Big) +  \frac{2 \rho(t) \eta (t)}{M} \sum\nolimits_{m=1}^{M} \sum\nolimits_{i=2}^{\tau} \left( F_m^* - \mathbb{E} \left[ F_m({\boldsymbol{\theta}}_m^i (t)) \right] \right) \nonumber\\
& \quad - \frac{\mu \rho (t) \eta(t)}{M} \sum\nolimits_{m=1}^{M} \sum\nolimits_{i=2}^{\tau} \mathbb{E} \left[ \left\| \boldsymbol{\theta}_m^i (t) - {\boldsymbol{\theta}}^* \right\|_2^2 \right],  
\end{align}
where (a) follows since $\mathbb{E}_{\xi} \left[ \nabla F_m \left( \boldsymbol{\theta} (t), \xi_m^i (t) \right) \right] = \nabla F_m \left( \boldsymbol{\theta} (t)  \right)$, $\forall i, m, t$, and (b) follows due to the fact that $F_m$ is $\mu$-strongly convex. We have
\begin{align}\label{EQ_LemmaTermE_app_2}
& - \left\| \boldsymbol{\theta}_m^i (t) - {\boldsymbol{\theta}}^* \right\|_2^2 = - \left\| \boldsymbol{\theta}_m^i (t) - {\boldsymbol{\theta}} (t) \right\|_2^2 - \left\| \boldsymbol{\theta} (t) - {\boldsymbol{\theta}}^* \right\|_2^2 - 2 \langle \boldsymbol{\theta}_m^i (t) - {\boldsymbol{\theta}} (t) , \boldsymbol{\theta} (t) - {\boldsymbol{\theta}}^* \rangle \nonumber\\
& \qquad \quad \firstinequal - \left\| \boldsymbol{\theta}_m^i (t) - {\boldsymbol{\theta}} (t) \right\|_2^2 - \left\| \boldsymbol{\theta} (t) - {\boldsymbol{\theta}}^* \right\|_2^2 + \frac{1}{\eta(t)} \left\| \boldsymbol{\theta}_m^i (t) - {\boldsymbol{\theta}} (t) \right\|_2^2 + \eta(t) \left\| \boldsymbol{\theta} (t) - {\boldsymbol{\theta}}^* \right\|_2^2 \nonumber\\
&\qquad \quad = - (1 - \eta(t)) \left\| \boldsymbol{\theta} (t) - {\boldsymbol{\theta}}^* \right\|_2^2 + \Big(\frac{1}{\eta(t)} - 1\Big) \left\| \boldsymbol{\theta}_m^i (t) - {\boldsymbol{\theta}} (t) \right\|_2^2,
\end{align}
where (a) follows from Cauchy-Schwarz inequality. Noting that $\Gamma =F^* - \frac{1}{M} \sum\nolimits_{m=1}^{M} F_m^*$, the proof of Lemma \ref{LemmaTermE} is completed by substituting the result in \eqref{EQ_LemmaTermE_app_2} into \eqref{EQ_LemmaTermE_app_1}.

\bibliographystyle{IEEEtran}
\bibliography{Report}

\end{document}